\documentclass{article}
\usepackage[utf8]{inputenc}

\usepackage{fullpage}
\usepackage{bbm}
\usepackage{graphicx}
\usepackage{upref}
\usepackage{enumerate}
\usepackage{latexsym}
\usepackage{color,graphics}
\usepackage{comment} 
\usepackage{relsize}
\usepackage{url} 
\usepackage{bm} 
\usepackage{marvosym}
\usepackage{thm-restate}

\usepackage[section]{algorithm} 
\usepackage{amsmath,amsthm,amssymb,algorithmic,epsfig,graphicx, tikz}

\usepackage{amsfonts}
\usepackage{varioref}
\usepackage{authblk}

\newtheorem{theorem}{Theorem}[section]
\newtheorem{Lemma}[theorem]{Lemma}
\newtheorem{claim}[theorem]{Claim}
\newtheorem{proposition}[theorem]{Proposition}
\newtheorem{corollary}[theorem]{Corollary}

\newtheorem{definition}[theorem]{Definition}

\newcommand{\R}{\mathbb{R}}

\def\bra#1{\mathinner{\langle{#1}|}}
\def\ket#1{\mathinner{|{#1}\rangle}}
\newcommand{\braket}[2]{\langle #1|#2\rangle}

\renewcommand{\part}[2]{\frac{\partial #1}{\partial #2}}

\newcommand{\all}[2]{\begin{align}\label{#2} #1\end{align}}
\newcommand{\al}[1]{\begin{align} #1\end{align}}

\newcommand{\enum}[1]{\begin{enumerate}#1\end{enumerate}}

\newcommand{\en}[1]{\left ( #1 \right )}

\newcommand{\nl}{\notag \\}

\newcommand{\norm}[1]{\lVert#1\rVert}

\newcommand{\thmref}[1]{\hyperref[#1]{{Theorem~\ref*{#1}}}}
\newcommand{\lemref}[1]{\hyperref[#1]{{Lemma ~\ref*{#1}}}}
\newcommand{\remref}[1]{\hyperref[#1]{{Remark~\ref*{#1}}}}
\newcommand{\corref}[1]{\hyperref[#1]{{Corollary~\ref*{#1}}}}
\newcommand{\eqnref}[1]{\hyperref[#1]{{Equation~(\ref*{#1})}}}
\newcommand{\claimref}[1]{\hyperref[#1]{{Claim~\ref*{#1}}}}
\newcommand{\remarkref}[1]{\hyperref[#1]{{Remark~\ref*{#1}}}}
\newcommand{\propref}[1]{\hyperref[#1]{{Proposition~\ref*{#1}}}}
\newcommand{\factref}[1]{\hyperref[#1]{{Fact~\ref*{#1}}}}
\newcommand{\defref}[1]{\hyperref[#1]{{Definition~\ref*{#1}}}}
\newcommand{\exampleref}[1]{\hyperref[#1]{{Example~\ref*{#1}}}}
\newcommand{\hypref}[1]{\hyperref[#1]{{Hypothesis~\ref*{#1}}}}
\newcommand{\secref}[1]{\hyperref[#1]{{Section~\ref*{#1}}}}
\newcommand{\chapref}[1]{\hyperref[#1]{{Chapter~\ref*{#1}}}}
\newcommand{\apref}[1]{\hyperref[#1]{{Appendix~\ref*{#1}}}}

\DeclareMathOperator{\poly}{poly}
\DeclareMathOperator{\polylog}{polylog}
\DeclareMathOperator{\diag}{diag}
\DeclareMathOperator{\QFT}{QFT}

\title{A quantum spectral method for simulating stochastic processes, with applications to Monte Carlo}
\author[1, 2]{Adam Bouland}
\author[1]{Aditi Dandapani}
\author[1]{Anupam Prakash}

\affil[1]{\footnotesize QC Ware, Palo Alto, CA, USA}
\affil[2]{\footnotesize Department of Computer Science, Stanford University, Stanford, CA, USA}

\date{\today}

\begin{document}

\clearpage\maketitle
\thispagestyle{empty}

\begin{abstract}
 Stochastic processes play a fundamental role in physics, mathematics, engineering and finance. 
One potential application of quantum computation is to better approximate properties of stochastic processes. For example, quantum algorithms for Monte Carlo estimation combine a quantum simulation of a stochastic process with amplitude estimation to improve mean estimation.
In this work we study quantum algorithms for simulating stochastic processes which are compatible with Monte Carlo methods. We introduce a new ``analog'' quantum representation of stochastic processes, in which the value of the process at time t is stored in the amplitude of the quantum state, enabling an exponentially efficient encoding of process trajectories. 
    We show that this representation allows for highly efficient quantum algorithms for simulating certain stochastic processes, using spectral properties of these processes combined with the quantum Fourier transform.
    In particular, we show that we can simulate $T$ timesteps of fractional Brownian motion using a quantum circuit with gate complexity  $\polylog(T)$, which coherently prepares the superposition over Brownian paths.
 We then show this can be combined with quantum mean estimation to create end to end algorithms for estimating certain time averages over processes in time $O(\polylog(T)\epsilon^{-c})$ where $3/2<c<2$ for certain variants of fractional Brownian motion, whereas classical Monte Carlo runs in time $O(T\epsilon^{-2})$ and quantum mean estimation in time $O(T\epsilon^{-1})$.
 Along the way we give an efficient algorithm to coherently load a quantum state with Gaussian amplitudes of differing variances, which may be of independent interest.

\end{abstract}

\clearpage
\pagenumbering{arabic} 
\newpage

\section{Introduction}

Stochastic processes play a fundamental role in mathematics, physics, engineering, and finance, modeling time varying quantities such as the motion of particles in a gas, the annual water levels of a reservoir, and the prices of stocks and other commodities. 
One potential application of quantum computation is to better estimate properties of stochastic processes or random variables derived therefrom.
A long line of works have shown that one can quadratically improve the precision of estimating expectation values of random variables, using quantum amplitude estimation and variants thereof \cite{abrams1999fast,heinrich2002quantum,heinrich2002optimal,heinrich2003monte,brassard2011optimal,montanaro2015quantum, hamoudi2018quantum,hamoudi2021quantum,kothari2022mean}.
For example, if one wishes to estimate the expectation value of a random variable that one can efficiently classically sample, then classical Monte Carlo methods require $\Theta(1/\epsilon^2)$ samples to $\epsilon$-approximate the mean, while
quantum algorithms can do so using only $\Theta(1/\epsilon)$ calls to the classical sampling algorithm in superposition.
These algorithms are optimal in a black-box setting \cite{bennett1997strengths,nayak1999quantum}. 
Such algorithmic approaches have received much attention as a potential application of quantum computation.
For example, there has been much excitement about the possibility of using this approach for the Monte Carlo pricing of financial derivatives and risk analysis, e.g. \cite{rebentrost2018quantum,woerner2019quantum,bouland2020prospects,egger2020quantum,stamatopoulos2020option,egger2020credit,chakrabarti2021threshold,doriguello2021quantum,doriguello2022quantum}.

However, achieving a practical quantum speedup for estimation of expectations over stochastic processes is challenging, even with potential future improvements in quantum hardware. 
This is for two reasons.
First, in these algorithms one must simulate the underlying random variable/stochastic process in quantum superposition.
That is, one needs to coherently prepare a quantum state which encodes the randomness used to generate the trajectory as well as a trajectory of the stochastic process under consideration.
While in principle this can always be done in the same amount of time to classically simulate the process -- for example by compiling the classical simulation down to Toffoli gates with uniform random seeds as input -- in practice this can result in prohibitively 
large gate counts in the simulation circuit.
Second, one must not only simulate the above simulation circuit once, but ($O(1/\epsilon)$) in series in order to achieve the quadratic quantum speedup\footnote{It is however possible to perform lower-depth variants of the algorithm at the cost of decreased speedups \cite{giurgica2022low} that are proportional to the number of times the simulation circuit is performed in series.}. The depth for the simulation circuit is therefore another bottleneck in obtaining speedups
for quantum Monte Carlo methods. 
Due to these constraints, it has recently been noted that in certain future projections of quantum hardware development, the clockspeed overheads of quantum error correction might overwhelm the quadratic speedups for relevant parameter regimes \cite{babbush2021focus,troyersimons}. For example, recent estimates of the effective error rate needed to implement financial derivative pricing in a practical setting using state of the art algorithms has revealed it might require many orders of magnitude improvements over existing hardware \cite{chakrabarti2021threshold}.

Fortunately, these quantum Monte Carlo algorithms naturally compound with any speedup in the process simulation.
Therefore, a critical goal is to find a quantum speedup for simulating stochastic processes, or at the very least a more gate-efficient method of simulating such processes, to render these techniques practical in the future. 
Indeed, \cite{montanaro2015quantum} noted that quantum walk methods can achieve such a speedup in certain cases, and used this to show a quantum algorithm for estimating the partition function of the Ising model exhibiting a quadratic speedup in both the error parameter $\epsilon^{-1}$ and the mixing time of the corresponding random walk over classical methods.
In a similar spirit there has been interest in efficient loading of particular probability distributions, such as the Gaussian distribution \cite{rattew2021efficient}, into quantum registers for future use in finance algorithms.

\subsection{Our results}

In this work we study the quantum simulation of stochastic processes for use in Monte Carlo algorithms. We focus on two questions: first, beyond quantum walks, are there scenarios can one create a coherent quantum simulation of a stochastic process using significantly fewer gates than trivially ``quantizing'' a classical simulation algorithm (i.e. compiling to Toffolis)? And second, could this create an end to end algorithm for an any potentially relevant applications which surpasses classical Monte Carlo?

We answer both questions in the affirmative.
First, we introduce a new notion of stochastic process simulation which we call the \emph{analog} simulation of a process, as opposed to the \emph{digital} simulation obtained by ``quantizing'' a classical algorithm.
We then show that one can create a highly efficient analog simulation for Brownian motion (BM) and a generalization thereof known as fractional Brownian motion (fBM). In particular we show how to $\epsilon$-approximately simulate a $T$-step Brownian motion process in merely $\tilde{O}(\polylog(T)\poly(\epsilon^{-1}))$ qubits and even shorter circuit depth. 

\begin{theorem} [Main Theorem, informal\footnote{See Theorem \ref{thm:FBM} for formal statement.}]
There is a quantum algorithm to produce an $\epsilon$-approximate analog simulation for fractional Brownian motion with Hurst parameter $H\in (0, 1]$, using a quantum circuit with 
$O(\polylog(T)+\poly(\epsilon^{- 1/2H}))$ gates,  $\widetilde{O}(\polylog(T)+\polylog(\epsilon^{-1/2H}))$ depth, and $O(\polylog(T)+O(\epsilon^{-1/2H}))$ qubits.
\end{theorem}

\noindent Here the input to the algorithm is a description of the parameters of the fractional Brownian motion -- namely the drift, variance and the ``Hurst parameter'' which describes the amount of correlation or anti-correlation between subsequent steps of the Brownian motion. We define this more formally in Section \ref{p1}.
Our algorithm makes critical use of the quantum Fourier transform and spectral properties of fractional Brownian motion, as well as a recursive application of data loading algorithms which uses special properties of this spectrum.
Additionally, we generalize these methods to a broader class of stochastic processes known as L\'{e}vy processes, albeit with weaker simulation guarantees.

Second, we show how to use this new representation to obtain an end to end quantum algorithm for estimating properties of stochastic processes, which is faster than classical Monte Carlo. This is not straightforward, as our analog simulation makes use of the exponential size of Hilbert space to efficiently encode the stochastic process trajectories. This does not allow one to directly measure quantities readily available in the digital simulation. For example, one cannot easily read out the value of the process at a particular time, similar to how the HHL algorithm does not allow one to extract individual entries of the solution vector of a linear system \cite{harrow2009quantum,aaronson2015read}.
Therefore, some work must be done to identify properties of stochastic processes which are easily extractable from this analog representation. 

To this end, we describe two quantities which are \emph{time averages} of the stochastic processes which meet this criteria, which can be efficiently estimated by combining our analog simulation algorithm with quantum mean estimation algorithms \cite{brassard2011optimal,montanaro2015quantum, hamoudi2021quantum,kothari2022mean}. 
For example, we show that one an efficiently price a certain over-the counter financial option currently traded -- in particular, an option on realized variance -- under a particular assumption about the evolution of the asset.
We also show that one can create an efficient statistical test for anomalous diffusion in fluids.
The first algorithm runs in time $O(\polylog(T)\epsilon^{-c})$ where $3/2<c<2$ is a constant depending on certain parameters of the stochastic process. This is an improvement over classical Monte Carlo which runs in time $O(T\epsilon^{-2})$, and incomparable\footnote{Here the algorithms are incomparabale as coming from the fact that our analog simulation introduces additional error terms into the simulation at order $\poly(\epsilon^{-1})$, which are exponentially suppressed in the digital simulation. } to standard quantum mean estimation which runs in time $O(T\epsilon^{-1})$. 
Our algorithm therefore creates a black-box quantum speedup compared to the best black-box classical algorithm.

We leave open the question of whether our techniques can generate a genuine (white box) quantum speedup for estimating certain properties of stochastic processes.
Here the central questions are a) to characterize what sorts of properties of fBM we can estimate with our methods and b) to determine if there exist faster classical methods for computing such properties than classical Monte Carlo sampling.
For the particular quantities we consider here, there exist closed-form analytical formulae for these quantities, and therefore our results do not represent white-box speedups over the best possible classical algorithm.
However, our technique easily generalize to (mildly) postselected subsets of the process trajectories, which quickly allow one to depart from the regime of closed-form analytical formulae. 
Therefore we expect that our techniques can easily price certain options or evaluate properties of sub/super diffusive fluids which do not have closed-form analytical formulae, and therefore could possibly represent white-box quantum speedups.
As with all black-box speedups (including standard quantum mean estimation algorithms), the central question is whether or not faster classical algorithms exist beyond classical Monte Carlo despite the non-existence of closed form solutions.
We discuss this further, as well as additional open problems, in Section \ref{subsec:openqs}.

\subsection{Analog vs digital simulation}

A discrete-time stochastic process $S(T)$ is a description of a probability distribution over values $v_1, v_2,\ldots v_T$ of a particular quantity at specified times $1,2,\ldots T$. 
The differences between successive values $v_{i+1}-v_i$ are referred to as the increments of the processes, which might be dependent on one another.
For example, the  stochastic process describing a particle subject to diffusion will have independent increments, while the process representing the annual water level in a reservoir will have positively correlated increments due to long-term drought cycles \cite{H66}. 

We say one can perform a \emph{quantum simulation} of a stochastic processes if one can coherently produce a quantum state $\psi$ representing the stochastic process. The complexity of this task might depend on the quantum representation of the stochastic process. For example, a commonly used representation (e.g. as mentioned in \cite{kothari2022mean}) is to consider the quantum state
\[\displaystyle\sum_{v_1,v_2,\ldots v_T} \sqrt{p_{v_1,v_2,\ldots v_T}} \ket{v_1,v_2,\ldots v_T} \ket{g} \]
where $p_{v_1,v_2,\ldots v_T}$ is the probability the process takes values $v_1,v_2,\ldots v_T$, $\ket{g}$ is a garbage state entangled with the values, and the sum is taken over all possible values of the tuples $v_1\ldots v_T$.
In other words, the ket of the state encodes the trajectory of the process (i.e. the tuple of values $v_1,v_2,\ldots v_T$), and the amplitude stores the probability that trajectory occurs.
If one traces out the garbage qubits, the diagonal entries of the reduced density matrix are precisely the probability distribution of the stochastic process.

We call this the \emph{digital representation} of the stochastic process, because  it meshes well with classical digital simulation algorithms.
Namely, if one has a classical algorithm to sample from $v_1,v_2,\ldots v_T$ in time $f(T)$, then it immediately implies a quantum algorithm to produce the digital representation in time $O(f(T))$ -- simply by compiling the algorithm down to Toffoli gates, and replacing its coin flips with $\ket{+}$ states\footnote{The coin flip registers then become the garbage register of the above state.}.
Therefore there is no quantum \emph{slowdown} for the digital simulation task in general.
To the best of our knowledge, the best type of quantum \emph{speedup} for digital simulation occurs via quantum walk algorithms as noted in \cite{montanaro2015quantum}, where $f(T)$ goes to $\sqrt{f(T)}$ in certain cases.

In this work we introduce a new representation of stochastic processes, which we call the \emph{analog} representation. 
We first describe the analog encoding of a single trajectory $ (v_1,v_2,\ldots v_T)$ of the stochastic process $S(T)$ is defined as,

\[ \ \ket{\psi_{v_1,v_2,\ldots v_T}}=\left(\frac{1}{\norm{v}}\displaystyle\sum_{i=1}^T v_{i}\ket{i}\right) \]

\noindent Here the value of the process is encoded in the \emph{amplitude} of the quantum state and the ket stores the time. This representation of the stochastic process manifestly takes advantage of the exponential nature of quantum states -- as only $O(\log T)$ qubits are required to represent a $T$-timestep process. It is an analog representation as the values are stored in the amplitude of the state, rather than digitally in the value of the ket.
An $\epsilon$-approximate encoding of the trajectory of the stochastic process trajectory is a state such that $\norm { \ket{\psi'_{v_1,v_2,\ldots v_T}} - \ket{\psi_{v_1,v_2,\ldots v_T}} }^{2}  \leq \epsilon$. 
Note that here the representation discards the normalization information, but we will later consider modifications of this formalism which keeps normalization information as well, at the cost of introducing an additional flag register which is 0 on the desired (sub-normalized) state.

Preparing the analog encoding of a single trajectory of $S(T)$ is equivalent to the task of preparing copies of a density matrix $\rho= \mathbb{E}_{v_1,v_2,\ldots v_T} \left[ \ket{\psi_{v_1,v_2,\ldots v_T}}\bra{\psi_{v_1,v_2,\ldots v_T}}\right] $. Such a density matrix represents a single trajectory of the stochastic process sampled according to the correct probabilities.  
However, for quantum Monte Carlo methods to estimate a function of a stochastic process, a stronger notion of coherent analog encodings for $S(T)$ is required.

The coherent analog representation of the stochastic process $S(T)$ is defined as follows
\[ \ket{S(T)}= \displaystyle\sum_{v_1,v_2,\ldots v_T} \sqrt{p_{v_1,v_2,\ldots v_T}} \ket{\psi_{v_1,v_2,\ldots v_T}} \ket{g}  \]
where $\ket{g}$ is an orthonormal garbage register entangled with the trajectory $v_1\ldots v_T$.
In other words, we assume that tracing out the garbage register yields the state
\[\rho = \displaystyle\mathbb{E}_{v_1,v_2,\ldots v_T \sim S} \ket{\psi_{v_1,v_2,\ldots v_T}}\bra{\psi_{v_1,v_2,\ldots v_T}} .\]
The garbage register essentially encodes the randomness needed to sample from the process trajectories.
The coherent analog representation $\ket{S(T)}$ of the stochastic process is compatible with quantum Monte Carlo methods and can be used as part of the simulation circuit/oracle for estimating a function of the stochastic process using the 
quantum amplitude estimation algorithm. An $\epsilon$ approximate analog encoding for $\ket{S_{\epsilon}(T)}$ is defined similarly where the trajectories generated are $\epsilon$ approximately correct. We note that a method of preparing $\rho$ directly (e.g. by classically sampling random trajectories and then coherently preparing the trajectory states) is not compatible with amplitude estimation as it is not unitary.

\subsection{Proof sketch}
Our first result is to show that the mathematical structure of fractional Brownian motion -- a fundamental stochastic process that can be used to model diffusion processes like the motion of particles in a gas -- is particularly amenable to efficient analog simulation.

\subsubsection{Step 1: View in Fourier basis using the QFT}
Our algorithm is derived from combining spectral techniques with the quantum Fourier transform. 
The starting point is the classic spectral analysis of Brownian motion and its Wiener series representation as a Fourier series with stochastic coefficients.
Brownian motion is a continuous time process, i.e. a probability distribution over continuous real-valued functions $B(t):[0,1]\rightarrow \mathbb{R}$, with Gaussian increments between distinct times. 
There are several mathematical definitions of Brownian motion, but the most helpful is a description of its Fourier series due to Wiener (1924).
Wiener observed that Brownian motion\footnote{Technically, this is the description of a Brownian bridge which has fixed start and end points. But this spectral analysis can be generalized to other forms of Brownian motion.} with zero drift and variance $\sigma$ on the interval $[0,1]$ can be written as
\[B(t) =  \displaystyle\sum_{k=1}^{\infty} \frac{a_k}{k} \sin(\pi kt)\]
where the variables $a_k$ are independent identically distributed (i.i.d.) Gaussian variables with mean 0 and variance 1.
In other words, when viewed in frequency space, Brownian motion is extremely simple -- its frequency components decouple from one another.

The Wiener series representation gives rise to a family of classical spectral algorithms for simulating Brownian motion, which are commonly used in computational finance \cite{cherubini2010fourier}. The basic idea is to discretize time to $T$ timesteps, draw a set of random Gaussian variables $a_k$ of diminishing variances (typically imposing some cutoff on the maximum frequency considered), and take their Fourier transform to obtain a trajectory of $B(t)$.
Classically this takes time $O(T\log T)$ via the fast Fourier transform algorithm \cite{cooley1965algorithm}.

Our first key observation is that this classical spectral algorithm offers an opportunity for a highly efficient quantum analog simulation for Brownian motion trajectories via the quantum Fourier transform (QFT). 
The QFT performs a Fourier transform over a vector of length $T$ using only $\polylog(T)$ qubits and quantum gates -- essentially by exponentially parallelizing the FFT algorithm -- and is at the is at the core of many quantum speedups, e.g. \cite{shor1999polynomial}.
Therefore, if one could efficiently prepare a quantum state encoding the Fourier transform of a stochastic process, then by taking its QFT\footnote{For technical reasons we perform a Real version of quantum Fourier transform on such a vector known as the Discrete Sine Transform (DST), but this also admits an efficient quantum circuit implementation based on the QFT circuit.} one would obtain an analog simulation of the process.
The problem of analog simulation therefore reduces to the problem of loading the stochastic coefficients of the processes' Fourier transform.
This is the conceptual core of our quantum spectral algorithm.

For Brownian motion we therefore need to efficiently load a quantum state where the amplitudes are distributed as independent Gaussians of diminishing variances\footnote{We note that this task is different than the one considered in \cite{rattew2021efficient}, as we are probabilistically loading the Gaussian into a single amplitude, rather than deterministically preparing a single quantum state with many amplitudes in the shape of a Gaussian.}.
As we discuss next, the symmetries of the Brownian motion and the decoupling of the stochastic coefficients allows us to obtain a very efficient quantum analog 
simulation algorithm for the Brownian motion. A similar analysis holds for fractional Brownian motion as well. Here the Fourier coefficients of the process decouple as well to independent Gaussians, but the functional form of the diminishing variances is given by a power low function of the Hurst parameter which controls the amount of correlation between steps. For simplicity of presentation, we will sketch our algorithm for standard Brownian motion. The extension to fractional Brownian motion will later be shown in Section \ref{efbm} using similar ideas.

\subsubsection{Step 2: Efficient Gaussian loading}

Via the QFT, we have shown that to produce an analog simulation of a single trajectory of Brownian motion, we need to show how to efficiently prepare a quantum state encoding its Fourier transform.
In other words, we need to prepare the state,

\[\ket{\phi_{\vec{a}}}\propto \displaystyle \sum_{k=1}^{T} \frac{a_k}{k} \ket{k}\]
where the $a_k \sim N(0,1)$ and the series coefficients are independent Gaussian variables of decreasing variance according the function $f(k)=1/k$.
We call this the ``Gaussian loading problem'' for the function $f(k)=1/k$ -- and in general one can consider this problem with different decay functions of the variance.

The first step of our algorithm is to truncate the Fourier series to a finite number of terms $L$. That is, we instead prepare the state
\[\ket{\phi_{\vec{a}}}\propto \displaystyle \sum_{k=1}^{L} \frac{a_k}{k} \ket{k}\]
This introduces a small amount of error in our simulation algorithm. However, as the function $1/k$ is rapidly diminishing as a function of $k$, we show that this only introduces a small amount of error in our simulation. More generally in one wishes to find an $\epsilon$-approximate simulation algorithm, this only requires setting $L=\poly(\epsilon^{-1})$.

We then give an efficient quantum for solving this truncated Gaussian loading problem, which we believe may be of independent interest.
Our algorithm uses only $O(L+\log T + \log(\epsilon^{-1}))$ qubits and computation time. 
Our algorithm applies to a variety of decay functions for the variance -- which will play a key role in our generalization to fractional Brownian motion.
This algorithm is the technical core of our results.

The starting point for our Gaussian loading algorithm is highly efficient data loader circuits \cite{J21} for particular quantum states. These are circuit realizations of previous recursive constructions that used specialized quantum memory devices such as those of Grover and Rudolph \cite{grover2002creating} and Kerenidis and Prakash \cite{kerenidis2016quantum}. For any fixed values of the $a_i$, one can define a log-depth circuit to load the vector $\ket{\phi_{\vec{a}}}$, by now standard recursive doubling tricks -- one simply computes how much $\ell_2$ mass is on the first vs second half of the state, hard-codes this as a rotation angle between the first and second halves, and recurses in superposition.
This results in a log-depth circuit for loading the state, where $k$ is represented in unary. 

There are two issues which must be solved to apply this algorithm to our Gaussian loading problem.
For one, this loading occurs in unary, and we are using a binary representations of $k$ and $T$ in our analog encoding, but this turns out to be a minor issue which can be solved with low-depth binary to unary converters (see Appendix \ref{app:c}). The second and more fundamental issue is that this only describes how to efficiently load a single state, and we wish to load the analog representation of the Brownian motion $\ket{B(t)}$ in order to be 
compatible with quantum Monte Carlo methods. 

We solve this by applying the data loading algorithm twice recursively -- which we call ``data loading the data loader." The basic idea is that for any data loading algorithm $\mathcal{A}$, it takes as input some rotation angles $\vec{\theta}$, and outputs a state $\ket{\mathcal{A}(\vec{\theta})}$.
The data loading algorithms we consider are \emph{onto}, in other words for any state $\ket{\psi}$, there exists a setting of the angles $\ket{\theta}$ such that $\mathcal{A}(\vec{\theta})=\ket{\psi}$.
Thus, given any distribution $D$ over quantum states, this induces a classical probability distribution $D'$ over vectors $\vec{\theta}$. Therefore, if we could only efficiently load the quantum state corresponding to $D'$, i.e.
\[\ket{D'} = \displaystyle\sum_{\vec{\theta}}\sqrt{D'(\vec{\theta})}\ket{\vec{\theta}}\]
where $D'(\vec{\theta})$ is the probability of $\vec{\theta}$ in $D'$, 
then by feeding this state into $\mathcal{A}$ and tracing out the angle registers, this would efficiently allow us to prepare $\sigma$.

If one considers applying this technique directly, however, it turns out to produce highly complex quantum circuits. 
While there exists a distribution $D'$ over data loader angles to produce states of independent diminishing Gaussians, the joint distribution induced on angles is quite complicated.
In particular, the angles are highly \emph{correlated} with one another.
In other words, the induced distribution on the angle $\theta_i$ applied at a particular stage of the algorithm is dependent on the prior angles applied.
Therefore, to load the probability distribution on angles $D'$ would be prohibitively costly -- as it would require solving a highly correlated data loading problem across many qubit registers. This increases the complexity of the data loading circuit which reduces or eliminates our potential advantage from using the QFT.

We circumvent this obstacle in two steps. First, we show one can highly efficiently load large vectors of i.i.d. Gaussians, i.e. where the Gaussian entries all have the same variance. 
This is because the induced distribution on data loader angles is independent -- we show the angles decouple due to symmetries of the high-dimensional Gaussian, which mesh particularly well with the binary tree data-loading circuits.
In fact the distribution on data loader angles $\vec{\theta}$ turn out to have a closed form given in terms of the $\beta$ and $\gamma$ distributions due to the fact that the sum of squares of $k$ i.i.d. Gaussians are distributed according to the $\gamma(k/2)$ distribution.
Therefore, there is a highly efficient circuit to load these angles - one just loads each angle register separately, and feeds it into the \cite{J21} circuits.
This allows us to quickly prepare quantum states with i.i.d. Gaussian entries. 

Second, we show that we can efficiently \emph{convert} such states into states with decreasing variances with a simple trick.
The basic idea is to artificially inject diminishing variances with reversible addition -- we first prepare prepare $\displaystyle\sum_{1}^{L} b_{k} \ket{k}$ where the $b_{k}\sim N(0,1)$ as described above, then prepare the state $\displaystyle\sum_{1}^L \frac{1}{k} \ket{k}$, and reversibly add their register mod $L$ to produce
\[\displaystyle\sum_{k,k'=1}^{L} \frac{b_{k}}{k'} \ket{k}\ket{k'}\ket{k+k' \mod L}
\]
we then measure the auxiliary register to get the value of $l=(k+k' \mod L)$.
Post measurement, the amplitudes are $\frac{b_{k+l}}{k'}$ -- which looks like exactly what we desire, except the entries of the Guassian vector $b$ have been permuted (shifted by $l$). 
The key observation, however, is that the iid Gaussian vector $\vec{b}$ is permutation invariant.
Therefore
 $\frac{b_{k'+l}}{k'}$ ranging over $k'$ have the same distribution as i.i.d. Gaussian amplitudes with decaying variances, irrespective of the value of $l$. 
 The same technique works for a wide variety of functional forms of diminishing variances, which we discuss in detail in the main text, as it is key to generalizing our algorithm to fractional Brownian motion.

\subsubsection{Sketch of end to end applications}

We also provide two end-to-end examples using our analog encoding which provide black-box speedups over classical Monte Carlo sampling.

The first example entails the pricing of a variance swap. A variance swap is a financial instrument which pays out proportional to the mean squared volatility observed in a stock price -- so the more volatile the stock, the more it pays out. It can be used as a hedge against volatility, and is traded as an over the counter option in financial markets. We show that we can use our analog encoding to efficiently price a variance swap in certain conditions. At a high level, our algorithm is efficient because the price of a variance swap is naturally a \emph{time average} of a square of values of the volatility, and it is particularly easy to extract time averages from our analog encoding (as they are certain amplitudes of our state which are relatively large). Our algorithm works under a particular assumption about the time evolution of the price of the underlying asset. In particular, the asset must evolve by Geometric Brownian motion with changing variance, and where the variance evolves by fractional Brownian motion. By combining our algorithm with quantum mean estimation, we can $\epsilon$-approximate the value of the option in time $O(\polylog(T)\epsilon^{-c})$ where $3/2<c=1+1/2H<2$ where $H$ is the Hurst parameter of the fBM. Our algorithm works better with higher Hurst parameters of the volatility. See Section \ref{QMC} for details. 
While this particular option has a closed-form analytical solution, we note we can easily apply post-selection on the fBM, for example over paths whose norm lies in a given interval, which would circumvent the possibility of an analytic solution. 

Our second example involves a statistical test to distinguish between different diffusive regimes in single- particle motion. While in ideal fluids particle positions evolve by Brownian motion, in particular sub or super-diffusive fluids, they evolve by fBM with a nontrivial Hurst parameter.
We show that we can use our analog encoding for fBM to create a statistical test that distinguishes between a particle following an fBM with given Hurst parameter, or fBM with an alternative Hurst parameter, or even a simple case of a continuous-time-random walk. 
Our algorithm again runs in time $\tilde{O}(\polylog(T) \epsilon^{-c'}),$ where $c'>2$ depends on the Hurst parameter of the fBM, rather than $\tilde{O}(\poly(T)*N)$ in the classical case, where $N$ is determined by the discretization used to sample its characteristic function. 
Unlike our prior application, here we do not run quantum mean estimation, but rather use our simulation directly to produce estimates of the average mean-squared-displacement of the particle under certain Hurst parameters, which is compared to the observed data. Therefore this application has a worse scaling in $\epsilon^{-1}$ compared to classical, but a better scaling in $T$. As we will discuss shortly, this could still possibly generate a faster black-box algorithm than classical Monte Carlo in situations with large $T$.

\subsection{Generalizations and Open Problems}
\label{subsec:openqs}

There are many open problems remaining.
Of course, the most direct one is whether or not our method can produce an end to end asymptotic speedup for computing properties of certain stochastic processes, as previously discussed. 
Here the basic issue is to identify interesting properties of fractional Brownian motion which are complicated enough to require classical Monte Carlo approaches.
In this direction we believe considering functions of postselected subsets of trajectories is the most promising approach. 
Postselection typically takes one out of the regime of analytical formulae.
For example, in computational finance, barrier options (which only can pay out if the price of the underlying asset breaches a certain value at some point in time) typically do not have analytical formulae and therefore are priced by Monte Carlo methods. 
Postselection slows both our quantum algorithm and classical Monte Carlo in unison, preserving the relative speedup of our method relatively to classical MC.
Therefore, if one could find a postelected property of fBM for which the best classical estimation algorithm is classical MC, then this could yield a white-box speedup for our algorithm.

Another possible direction to search for speedups is to consider other inner products one could compute with respect to Brownian Motion. We generalize our algorithm to the following: given a function $f(t)$, one can efficiently evaluate its inner product with BM, i.e. $|\braket{f(T)}{B(T)}|^2$, assuming that one can prepare a quantum state encoding $f(t)$ (for details see Section \ref{sec:mcmethods}). 
Our given applications are the special case where $f$ is the indicator function between times $t_1$ and $t_2$.
One can ask if other functions might give a quantum speedup.

In any case, quantifying such a speedup would require careful work.
For one, there are many parameters at play.
To quantify an end to end speedup, one would need to take into account that the parameter $T$ is implicitly a function of $\epsilon$ (see e.g. \cite{montanaro2016quantum}).
For example if one must set $T=O(\epsilon^{-1})$ vs $T=O(\epsilon^{-2})$ vs $T=O(\epsilon^{-1/2})$, then our method's black-box speedup becomes polynomial, but with differing degrees.
We note that even if $T=O(\epsilon^{-1})$, our method's savings in $T$ pushes our algorithms' performance beyond that of standard quantum mean estimation -- and the gap only grows if $T$ is larger.  
Other factors might also affect the apparent speedup -- for example if the quantity being estimated is invariant to high-frequency components of the stochastic process (as with a time average), then our algorithms' omission of high-frequency content might result in a better error scaling than our naive bounds, and potentially result in a $\tilde{O}(\polylog(T) \epsilon^{-c})$ algorithm where $c<1$.
For this reason we believe quantifying potential asymptotic speedups for potential problems of interest to be an interesting line of inquiry.
More broadly, we leave open the question of whether our techniques can be ``de-quantized'' in a similar spirit to \cite{tang2019quantum,gilyen2018quantum,tang2021quantum,chia2022sampling,gharibian2022dequantizing}, i.e. if it is possible achieve a similar $\polylog(T)$ dependence for sampling from values/times of fBM trajectories\footnote{However, we note that this would only ``de-quantize'' the applications which do not make use of quantum mean estimation/amplitude estimation, as the latter algorithms intrinsically require coherent state preparation and not probabilistic samples.
}.

Another interesting direction is to explore what other stochastic processes might be amenable to efficient analog simulation. For example, would it be possible to give an efficient analog simulation for Geometric Brownian motion? The case of Brownian motion is particularly nice because its Fourier spectrum decouples.
 For a general stochastic processes, there are non trivial dependencies between the stochastic coefficients and the resources 
required for loading the joint distribution of the Fourier coefficients may be prohibitive. 
However there are more general families of stochastic processes with well-behaved spectra,
One reason the Fourier spectrum of Brownian motion is well-behaved it that it is a stationary process, i.e. the joint probability distribution does not change when shifted in time. This cyclic shift symmetry is precisely the symmetry of the QFT, and therefore it might be possible to give simulations for other stationary processes. In another direction, our results use the fact that Brownian motion can be expressed as an integral over white noise -- which is noise with i.i.d Gaussian Fourier components.

In this spirit, in Section \ref{sec:integrals} we generalize our algorithm to produce analog encodings of trajectories of L\'{e}vy processes.
L\'{e}vy processes are stationary stochastic processes generalizing Brownian motion, which can be expressed as integrals over a linear combination of Poisson and Brownian noise. The quantum simulation method for L\'{e}vy process trajectories is obtained by quantizing the classical method that embeds the Toeplitz discrete integration operator into a circulant matrix \cite{dieker2003spectral}. The method remains efficient in the quantum setting with gate complexity $O(\poly(\log T, 1/\epsilon))$ as circulant matrices are diagonalized by the quantum Fourier transform and further the Fourier spectrum of L\'{e}vy noise is flat, similar to the spectrum of the white noise. However our simulation results for L\'{e}vy processes are weaker than those for fBM, as we can only provide an incoherent simulation of these processes due to the coupling of the Fourier coefficients, and therefore cannot combine this method with amplitude estimation. The stochastic integral method can be used to generate encodings of It\^{o} processes that are defined as integrals over white noise and time.  We leave open the question of whether the quantum spectral method can be generalized further to (fractional) integrals over white noise and Poisson shot noise-- this  family of stochastic processes includes L\'{e}vy and It\^{o} processes as well as fractional Brownian motions. Indeed our extension to fractional Brownian motion -- which can be expressed as a fractional integral of Brownian motion -- is a step in this direction.

There is also the more direct question of whether analog simulation results in smaller quantum circuits than digital simulation for stochastic processes of interest beyond Brownian motion, in a non-asymptotic setting relevant to potential future applications of error-corrected quantum computers. This was part of our original motivation for this line of work, and we hope our work spurs further efforts in this area.

\section{Preliminaries} 
We introduce some preliminaries on stochastic processes and quantum computing in this section. 
In subsection \ref{p1}, we begin with the defining the Brownian motion. More generally, our techniques are applicable to stochastic processes that can be written as stochastic integrals over time and over Brownian motion, these processes include the fractional Brownian motion and It\^{o} processes. 
Subsection \ref{p2} introduces the quantum Fourier transform and state preparation circuits that will be used for constructing the quantum encoding for the stochastic processes. 
\subsection{Brownian motion and It\^{o} processes} \label{p1} 
The stochastic processes considered in this work are Brownian motion and its generalization to It\^{o} processes. We first introduce the Brownian motion and then the more general processes that can be represented as 
integrals over Brownian motion.  The Brownian motion is defined as follows, 
\begin{definition} 
The Brownian motion is a stochastic process $B: \R^{+} \to \R$ such that: 
\enum{ 
\item $B(0)=0$ and $B(t+h) - B(t) \sim N(0, h)$ for all $t, h \in \R^{+}$ where $N(0,h)$ is the normal distribution with mean $0$ and variance $h$. 
\item For any finite subdivision $0< t_{1} <t_{2} \cdots <t_{n}$ of time steps, the increments $B(t_{i+1}) - B(t_{i})$ are independent random variables. 
\item Almost surely, $t \to B(t)$ is a continuous function. 
} 
\end{definition} 
\noindent The existence of Brownian motion is not obvious from this definition, it was first demonstrated by Wiener \cite{W23} and 
a different alternate construction was given by L\'{e}vy \cite{MP10}.  Wiener obtained a stochastic Fourier series representation for the Brownian motion on $[0,\pi]$, 
\begin{theorem} (Wiener \cite{W23}) \label{wiener} 
For independent random variables $a_{i} \sim N(0,1)$, the following Fourier series represents Brownian motion on the interval $[0, \pi]$, 
\al{ 
B(t) = \sqrt{ \frac{1}{ \pi } }  a_{0} t + \sqrt{ \frac{2} {\pi}}  \sum_{k\geq 1} a_{k} \frac{ \sin(kt) }{k}  
} 
\end{theorem} 
\noindent The Brownian motion on $\R^{+}$ is obtained by concatenating independent Brownian motions on intervals $[k\pi, (k+1)\pi]$. Discarding the drift term in the Wiener series, one obtains the Brownian bridge, which represents a Brownian path with the values at the start and end point fixed to $0$. The Fourier series representation of the Brownian motion will also be used for the quantum simulation algorithm. 

More general stochastic processes can be defined as stochastic integrals over the Brownian motion. The stochastic integral $\int^{t}_{0} f(s) dB_{s}=  \lim_{N \to \infty} \sum_{i \in [N]} f(t_{i}) (B(t_{i}) - B(t_{t-1}) $ is defined as the limit over equal subdivisions $t_{i}, i \in [N]$ of the interval $[0,t]$ of the sum $\sum_{i \in [N]} f(t_{i}) (B(t_{i}) - B(t_{t-1})$. An important class of processes defined as stochastic integrals over Brownian motion is the fractional Brownian motion (fBM), a one parameter extension of Brownian motion for a Hurst parameter $H\in [0,1]$. The fBM with Hurst parameter $H=1/2$ corresponds to standard Brownian motion.  The fractional Brownian motion was first discussed by L\'{e}vy \cite{L53} as an integral over the standard Brownian motion.

\begin{definition} \label{levyfBM} 
The fractional Brownian motion with Hurst exponent $H \in [0, 1]$ is defined to be the stochastic process, 
\all{ 
fBM_{H} (t) := \int^{t}_{0} ( t- s)^{H-0.5} dB_{s} 
} {eq2} 
\end{definition} 
\noindent   Mandelbrot and Van Ness \cite{MV68} provided an origin independent fractional Brownian motion given by the 
Weyl integral, this definition as well can be written in integral form as  $\int^{t}_{0} K_{H}(s-t) dB_{s}$ for a Kernel function depending only on the difference $(s-t)$. The definition of fBM used for spectral simulation is that as a fractional integral of the white noise, the BM in turn can be viewed as  
integral of the white noise. 

Stochastic processes that can written as a linear combination of a stochastic integral over Brownian motion and a stochastic integral over time are called It\^{o} processes. An It\^{o} process has a representation of the form, 
\al{
X_{t} = X_{0} + \int \sigma_{s} dB_{s} + \int \mu_{t} dt 
} 
In addition to fractional Brownian motion, we also develop quantum simulation methods for generating trajectories of It\^{o} and L\'{e}vy processes 
that can be represented as stochastic integrals in section \ref{sec:integrals}. 

\subsection{Quantum computing preliminaries} \label{p2} 
We introduce in this section the quantum Fourier transform and logarithmic depth state preparation circuits that are components of the quantum stochastic process simulation 
algorithm. The quantum Fourier transform is defined as follows, 
\begin{definition} 
The quantum Fourier transform ($\QFT$) is an $N$ dimensional unitary matrix $U$ with entries given by $(U)_{jk} = e^{2\pi i jk/N} = \omega^{jk}$ for $0\leq j, k \leq n$, where $\omega = e^{2 \pi i/N}$ is an $N$-th root of unity. 
\end{definition} 
\noindent The real and the imaginary part of the the quantum Fourier transform are known as the discrete cosine transform (DCT) and the discrete sine transforms (DST) respectively. It is well known that the quantum Fourier transform (QFT) can be implemented as a logarithmic depth circuit, an explicit implementation using Hadamard and phase gates is provided in Appendix \ref{dst}. It also follows that the logarithmic depth circuit for the QFT can be used to implement the discrete sine and cosine transforms. 

The second quantum computing primitive that we use are the logarithmic depth state preparation circuits in quantum machine learning termed as data loaders \cite{J21}. The data loader circuit is a parametrized circuit that prepares the amplitude encoding $\ket{x}$ for a vector $x \in \R^{n}$. The data loader circuits are composed of recursive beam splitter (RBS) gates that are two qubit gates given as, 
\begin{equation} 
RBS(\theta) =  \begin{pmatrix} 
&1  &0 & 0 & 0 \\  
&0  & \cos(\theta) &\sin(\theta)  &0 \\
&0 & -\sin(\theta)   & \cos(\theta)   & 0 \\
&0 & 0 & 0     & 1 
\end{pmatrix}. 
\label{btheta} 
\end{equation}  

The logarithmic depth data loader circuit is illustrated in Figure \ref{fig1}, it outputs the state $\ket{x}$ on input $\ket{0^{n}}$. The data loader can be viewed as a circuit based realization of binary tree data structure for state preparation. 
The angles for the beam splitter gates in the data loader are determined by the vector $x$ that is being prepared by the circuit and are deterministic functions for quantum machine learning applications.

\begin{figure}[H] 
	\centering
	\includegraphics[scale = 0.6]{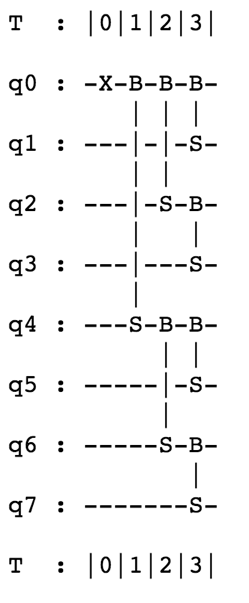}
	\caption{Quantum circuit for the 8 qubit unary data loader implemented using RBS gates \eqref{btheta}, labelled as B-S in the figure. } \label{fig1} 
\end{figure}

For the quantum simulation of stochastic processes, the data loader is used with stochastic input, that is the input vector $x$ is not fixed but drawn from a distribution over the unit sphere. The angles in the data loader circuit are thus drawn from a specific distribution, the explicit calculation of the angle distribution for the Haar random vector on the unit sphere will be an important part of the quantum algorithm for simulating Brownian motion trajectories. 

More explicitly, the angles for the data loader for vector $x \in \R^{n}$ are computed using a binary heap data structure where each node stores the sum of squares of the values in its subtree and an angle $\theta$. 
Denoting the value stored at node $j$ by $r(j)$ and, the angle $\theta$ is given by $\theta= arccos(\sqrt{\frac{r(2j)}{r(j)}})$ where $r(2j)$ is the sum of squares of the values stored in the left subtree for node $j$, 
that is $\cos^{2}(\theta) = \frac{r(2j)}{r(j)}$ and $\sin^{2}(\theta) = \frac{r(2j+1)}{r(j)}$. This description is useful for computing the distribution on the data loader angles for a Gaussian random vector. 

The data loader circuit produces a unary encoding for vector $x$ using $n$ qubits. It is further possible to convert the unary encoding into a binary encoding with a quantum circuit having depth poly-logarithmic in the dimension of the vector. The circuit for the unary to binary conversion is given in appendix \ref{app:c}.

\section{Quantum encodings of stochastic processes} 
A discrete-time stochastic process $S(T)$ is a description of a probability distribution over values $v_1, v_2,\ldots v_T$ of a particular quantity at specified times $1,2,\ldots T$. 
The differences between successive values $v_{i+1}-v_i$ are referred to as the increments of the processes, which might be dependent on one another. A \emph{quantum simulation} of a stochastic processes is a procedure to prepare quantum state $\psi$ representing the stochastic process. 

The complexity of this task depends on the quantum representation of the stochastic process. We recall first the commonly used digital quantum encoding for a stochastic process and 
then introduce two different types of analog encodings. 
\begin{definition}
A digital representation for a quantum stochastic processes is defined as the state, 
\[\displaystyle\sum_{v_1,v_2,\ldots v_T} \sqrt{p_{v_1,v_2,\ldots v_T}} \ket{v_1,v_2,\ldots v_T} \ket{g} \]
where $p_{v_1,v_2,\ldots v_T}$ is the probability the process takes values $v_1,v_2,\ldots v_T$ while $\ket{g}$ is a garbage state entangled with the values, and the sum is taken over all possible values of the tuples $v_1\ldots v_T$.
\end{definition} 
The registers in the digital encoding store the entire trajectory of the process (i.e. $(v_1,v_2,\ldots v_T)$), while the amplitude stores the probability that trajectory occurs.
If one traces out the garbage qubits, the diagonal entries of the reduced density matrix are precisely the probability distribution of the stochastic process.

We call this the \emph{digital representation} of the stochastic process, because  it meshes well with classical digital simulation algorithms.
Namely, if one has a classical algorithm to sample from $v_1,v_2,\ldots v_T$ in time $f(T)$, then it immediately implies a quantum algorithm to produce the digital representation in time $O(f(T))$ -- simply by compiling the algorithm down to Toffoli gates, and replacing its coin flips with $\ket{+}$ states. There is no quantum \emph{slowdown} for the digital simulation task, but to the best of our knowledge, nor are there any quantum speedups.

In this work we ask if quantum computation might admit faster algorithms for stochastic simulation tasks. 
We begin by introducing the \emph{analog} representation where the values of the stochastic process are stored in the amplitudes. 
\begin{definition} \label{d2} 
The analog encoding of a single trajectory $ (v_1,v_2,\ldots v_T)$ of the stochastic process $S(T)$ is the quantum state, 
\[ \ \ket{\psi_{v_1,v_2,\ldots v_T}}=\left(\frac{1}{\norm{v}}\displaystyle\sum_{i=1}^T v_{i}\ket{i}\ket{0}\right). \]
An $\epsilon$-approximate encoding of the trajectory of the stochastic process trajectory is a state such that $\norm { \ket{\psi'_{v_1,v_2,\ldots v_T}} - \ket{\psi_{v_1,v_2,\ldots v_T}} }^{2}  \leq \epsilon$. 
\end{definition}

\noindent Here the value of the process is encoded in the \emph{amplitude} of the quantum state and the ket stores the time. This representation of the stochastic process manifestly takes advantage of the exponential nature of quantum states -- as only $O(\log T)$ qubits are required to represent a $T$-timestep process. It is an analog representation as the values are stored in the amplitude of the state, rather than digitally in the value of the ket.

Preparing the analog encoding of a single trajectory of $S(T)$ is equivalent to the task of preparing copies of a density matrix $\rho= \mathbb{E}_{v_1,v_2,\ldots v_T} \left[ \ket{\psi_{v_1,v_2,\ldots v_T}}\bra{\psi_{v_1,v_2,\ldots v_T}}\right] $. Such a density matrix represents a single trajectory of the stochastic process sampled according to the correct probabilities.  
However, for quantum Monte Carlo methods to estimate a function of a stochastic process, a stronger notion of coherent analog encodings for $S(T)$ is required. 

\begin{definition}
The coherent analog representation of the stochastic process $S(T)$ is a superposition analog representation of the corresponding trajectories, 
\[ \ket{S(T)}= \displaystyle\sum_{v_1,v_2,\ldots v_T} \sqrt{p_{v_1,v_2,\ldots v_T}} \ket{\psi_{v_1,v_2,\ldots v_T}} \ket{g}  \]
Additionally, there is a garbage register $\ket{g}$ that encodes the randomness used to generate the corresponding trajectory.  An $\epsilon$ approximate analog encoding for $\ket{S_{\epsilon}(T)}$ is a superposition over $\epsilon$ approximate trajectories  $\sqrt{p_{v_1,v_2,\ldots v_T}} \ket{\psi'_{v_1,v_2,\ldots v_T}}$ with $\psi'_{v_1,v_2,\ldots v_T}$ being $\epsilon$ approximate trajectories as in definition \ref{d2}. 
\end{definition}

\section{Quantum simulation of Brownian Motion} 
The stochastic Fourier series representation of the Brownian motion (Theorem \ref{wiener}) provides a classical algorithm with complexity $O(T \log T)$ for simulating a Brownian path over $T$ steps using the fast Fourier transform. It also suggests a quantum algorithm for generating analog encodings of Brownian trajectories on a quantum computer. The quantum algorithm first prepares the state $\ket{W} = \sum_{i \in [L]} \frac{a_{i}}{i} \ket{i}$ obtained by truncating the Wiener series to a fixed number of terms $L$ using a data loader circuit and then applies the quantum Fourier transform circuit.

The number of terms $L$ required to obtain an approximate representation of the Brownian path are much smaller than the number of time steps, for example taking $L=200$ terms in the series leads to an $\ell_{2}$-norm error of $0.3$ percent. The quantum simulation algorithms use $O(L)$ gates and have complexity poly-logarithmic in $T$, and achieve a speedup over the classical simulator in the regime $T \gg L$,

We next describe an optimized algorithm for preparing analog encodings of Brownian motion with circuit depth $O(\log L) + \log (T))$, this algorithm will be used as a subroutine for generating the coherent analog encoding for fractional Brownian motions. 
\begin{theorem} \label{thm:BM} 
There is a quantum algorithm for generating $\epsilon$-approximate analog encodings of Brownian motion trajectories on $T$ time steps that requires $O(L + \log T)$ qubits, has circuit depth $\widetilde{O}( \log L + \log T)$ and gate complexity $O(L + polylog(T))$ for $L=O(1/\epsilon)$. 
\end{theorem} 
\noindent The poly-logarithmic gate complexity of the quantum Fourier transform further leads to the possibility of a potentially significant quantum speedup in for generating analog encodings of Brownian motion trajectories
 over $T$ time steps $T$. A classical algorithm that writes down the Brownian motion trajectory over $T$ time steps has complexity $O(T)$. 

The quantum algorithm for simulating Brownian paths is presented as Algorithm \ref{Brownian}. We describe the implementation of the different steps and then establish correctness of the algorithm.

\begin{algorithm} [H]
\begin{algorithmic}[1]
\REQUIRE Parameters $(L,T)$ where $L$ is the number of terms retained in the Wiener series and $T$ is the number of steps in the Brownian path. Both $L$ and $T$ are assumed to be powers of $2$. \\
Parameters $Z_{i}, 1\leq i\leq 3$ in the algorithm are appropriate normalizing factors. 
\ENSURE A quantum state representing the Brownian bridge $\ket{B_{L} } = \frac{1}{Z_{0}} \sum_{i \in [T]} B_{L} (i \pi /T) \ket{i}$ over $T$ time steps obtained by truncating the Wiener series to $L$ terms. 
\STATE  Prepare state $\ket{K}  = \frac{1}{Z_{1}} \sum_{k \in [L]} \frac{1}{k} \ket{k}$ using the unary data loader circuit. 
\STATE  Prepare state $\ket{R}  = \frac{1}{ Z_{2}} \sum_{i \in [L]} a_{i} \ket{i}$ where $a_{i}$ are independent $N(0,1)$ random variables using Lemma \ref{lem:gaussian}. 
\STATE  Starting with $\ket{K} \ket{R} = \frac{1}{ Z_{1}Z_{2}} \sum_{i, k \in [L] } \frac{a_{i}}{k} \ket{i, k}$, apply CNOT gates with first register as control to compute $(i \oplus k)$ in register 2 and measure register 2 in the standard basis to obtain 
$\frac{1}{ Z_{3}} \sum_{k \in [L] } \frac{ a_{k \oplus j } }{ k} \ket{k}  \ket{j}$. 
\STATE Trace out second register and apply a logarithmic depth unary to binary converter in Appendix \ref{app:c} to reduce the number of qubits to $O(\log L)$. Subsequently, pad with $0$ qubits so that the total number of qubits is $\log(T)$. 
\STATE The discrete sine transform matrix is the unitary matrix $U \in \R^{T \times T}$ with $U_{ij}= \sin( \pi i j/T)$. Apply the quantum circuit for $U$ from Corollary \ref{c0} to obtain $\ket{B_{L} } = \frac{1}{Z_{3}} \sum_{i \in [T]} B_{L} (i\pi/T) \ket{i}$. 
\end{algorithmic}
\caption{Quantum algorithm for generating analog encodings of Brownian paths.} 
\label{Brownian} 
\end{algorithm}
An efficient procedure for generating the random Gaussian state required for step 2 of algorithm \ref{Brownian} is described in Section \ref{rgs}, this procedure uses the data loader circuit but with angles drawn from a certain  
distribution to ensure that the vector prepared is the random Gaussian state. The discrete sine transform matrix can be implemented as a quantum circuit with depth $\widetilde{O}(\log T)$ as shown in appendix \ref{dst}. The total circuit depth for the algorithm is therefore $\widetilde{O}(\log L + \log T)$, the number of qubits used is $O(L+ \log T)$ and the gate complexity is $O(L + polylog(T))$ as claimed. The correctness of the Algorithm \ref{Brownian} is established in Lemma \ref{correct}. 
The dependence of the $\ell_{2}$ approximation error on the number of terms $L$ retained in the Wiener series is examined in Section \ref{error}.

\subsection{Gaussian state preparation} \label{rgs} 
In order to prepare the Gaussian state in step 2 of the algorithm, we need to compute the distribution on the angles for a unary data loader for a uniformly random unit vector in $S^{n}$ according to the Haar measure. Recall that a Haar random unit vector is obtained by choosing i.i.d. N(0,1) 
Gaussian random variables for the coordinates and rescaling to unit norm. Further, we show that for Haar random vectors, the angle distributions for the different angles in the data loader are independent and that these distributions can be specified exactly in terms of the gamma and beta distributions. We begin by recalling some of the useful properties of the beta and gamma distributions and then establishing the independence of the angle distributions for a uniformly random vector and explicitly computing the angle distributions. 

\begin{definition} 
The Gamma distribution $\gamma(a)$ has support $\R^{+}$, it is parametrized by $a>0$ and has cumulative distribution function, 
\al{ 
Pr [ X \leq x ] =  \frac{1}{ \Gamma(a)} \int^{x}_{0} t^{a-1} e^{-t} dt. 
} 
$\Gamma(a)$ is the Gamma function. 
 \end{definition} 
 \noindent The Beta distribution can be defined in terms of the Gamma distribution, 
 \begin{definition} \label{d1} 
 The Beta distribution $\beta(a,b)$ distribution is defined as $\frac{Y_{1}}{ Y_{1} + Y_{2} }$ where $Y_{1} \sim \gamma(a), Y_{2}\sim \gamma(b)$. 
 The density function for the $\beta(a,b)$ distribution is $\frac{\Gamma(a)\Gamma(b)}{\Gamma(a+b)} x^{\alpha-1}(1-x)^{\beta-1}$. 
 \end{definition} 
 A sum of squares of $k$ identically distributed $N(0,1)$ Gaussian random variables can be expressed in terms of the Gamma distribution.

\begin{proposition} \label{c1}
The sum of squares $\sum_{i} X_{i}^{2}/2$ where $X_{i}$ are $k$ independent Gaussian random variables has distribution $\gamma(k/2)$. 
\end{proposition} 
\noindent This fundamental fact underlies the 
 computation of the distribution for the data loader angles, a proof is provided in Appendix \ref{gamma} for completeness. 
 In addition to the above fact, we require a lemma that computes the distribution of $\theta$ if $\sin^{2}(\theta)$ has a $\beta(a, b)$ distribution. 
\begin{Lemma} \label{l6} 
If $\sin^{2}(\theta)$ is distributed according to $\beta(a,b)$, then $\theta$ has probability density function $F(t)=  \frac{\Gamma(a)\Gamma(b)}{\Gamma(a+b)}  \sin^{2a-1}(t) \cos^{2b-1}(t)$.  
\end{Lemma} 
\begin{proof} 
As $\sin^{2}(\theta)$ is distributed according to $\beta(a,b)$, we have that $\Pr[\sin^{2}(\theta) \leq  t] = \frac{\Gamma(a)\Gamma(b)}{\Gamma(a+b)} \int^{t}_{0} x^{a-1}(1-x)^{b-1} dx$.  
The function $\sin^{2}(\theta)$ is monotone on the interval $[0, \pi/2]$, 
\al{ 
\Pr[\theta \leq t] = \Pr[\sin^{2}(\theta)<\sin^{2}(t)] = \frac{\Gamma(a)\Gamma(b)}{\Gamma(a+b)} \int^{\sin^{2}(t)}_{0} x^{\alpha-1}(1-x)^{\beta-1} dx.  
} 
Substituting $x= \sin^{2}(z)$ so that $dx= 2\sin(z) \cos(z) dz$, the integral above reduces to, 
\al{ 
\frac{\Gamma(a)\Gamma(b)}{\Gamma(a+b)} \int^{t}_{0} \sin^{2a-1}(z) \cos^{2b-1}(z) dz.  
} 
It follows that the density function for $\theta$ is $F(t) = \frac{\Gamma(a)\Gamma(b)}{\Gamma(a+b)} \sin^{2a-1}(t) \cos^{2b-1}(t)$ as claimed. 
\end{proof} 
 
 \noindent Using the above probabilistic facts, we are now ready to compute the distribution of the data loader angles for a vector with coordinates given by i.i.d. Gaussian random variables.

\begin{Lemma} \label{l9} 
If the vector $x\in \R^{n}$ stored in the binary data loader is uniformly random, the angle $\theta$ at node of height $h$ has probability density function 
$\frac{\Gamma(2^{h-2})^{2}}{\Gamma(2^{h-1})} \sin^{2^{h-1}-1}(t) \cos^{2^{h-1}-1}(t)$. 
\end{Lemma} 
\begin{proof} 
The binary data loader for vector $x \in \R^{n}$ uses a binary heap data structure where each node stores the sum of squares of the values in its subtree and an angle $\theta$. 
Denoting the value stored at node $j$ by $r(j)$ and, the angle $\theta$ is given by $\theta= arccos(\sqrt{\frac{r(2j)}{r(j)}})$ where $r(2j)$ is the sum of squares of the values stored in the left subtree for node $j$, 
that is $\cos^{2}(\theta) = \frac{r(2j)}{r(j)}$ and $\sin^{2}(\theta) = \frac{r(2j+1)}{r(j)}$. 

If $x \in \R^{n}$ is a uniformly random vector then the values at the leaf nodes have distribution $X_{i}^{2}$ where $X_{i}$ are independent $N(0,1)$ random variables. The sum of squares $\frac{1}{2} \sum_{i \in [k]} X_{i}^{2}$ for a $k$ independent N(0,1) random variables $X_{i}$ has distribution $\gamma(k/2)$ by Proposition \ref{c1}. It follows from the definition of the beta distribution \ref{d1} that for a node at height $h$ (with the convention that the leaf nodes are at height $1$), $\sin^{2}(\theta)=  \frac{r(2j+1)}{r(j)}$ 
has distribution $\beta(2^{h-2}, 2^{h-2})$. Applying Lemma \ref{l6}, the angle $\theta$ for a node at height $h$ has probability density function $\frac{\Gamma(2^{h-2})^{2}}{\Gamma(2^{h-1})} \sin^{2^{h-1}-1}(t) \cos^{2^{h-1}-1}(t)$. 
\end{proof} 
We computed the distribution of the angles at the nodes for a  for a uniformly random vector the angles at the nodes of the binary data loader. The next Lemma shows that these angles are in fact independent for uniformly random vectors. 
\begin{Lemma} \label{l10} 
If the vector $x\in \R^{n}$ stored in the binary data loader is a Haar random unit vector, the angles $\theta_{i}, \theta_{j}$ stored at different nodes $i, j$ are independent. 
\end{Lemma} 
\begin{proof} 
If the nodes $i,j$ are at the same height the angles are independent as they are functions of different i.i.d. Gaussian random variables. Without loss of generality let the $h(i)\leq h(j)$ where $h$ is the height of the nodes. If $j$ 
does not lie on the path from the root to $i$, then the angles at $i,j$ are independent as they are functions of different i.i.d. Gaussian random variables. It therefore suffices to show that for a given node $j$, the angles in the 
left sub-tree rooted at $j$ are independent of the angle at $j$. 

The angles for the binary data loader are independent of the $\norm{x}$, that is the angles are the same for $x$ and $cx$ for all $c \in \R$. The angle at node $j$ is a function ratio of the norms of the left and right subtrees, 
that is $\theta= arctan(\sqrt{ r_{2j+1}/r_{2j}})$. It follows that for $i$ in the left-subtree rooted at $j$, the density function $f(\theta_{i}|\theta_{j})= f(\theta_{i} | r_{2j})= f(\theta_{i})$ establishing the independence of $\theta_{i}$ and $\theta_{j}$.

\end{proof} 
We are now ready to provide the procedure for generating the random Gaussian state in step 2 of Algorithm \ref{Brownian}. 
\begin{Lemma} \label{lem:gaussian}
The quantum state $\ket{R}  = \frac{1}{ Z_{2}} \sum_{i \in [L]} a_{i} \ket{i}$ where $a_{i}$ are independent $N(0,1)$ random variables can be prepared using the binary data loader 
construction with independent angles $\theta_{i}$ at height $h$ distributed according to the density function $\frac{\Gamma(2^{h-2})^{2}}{\Gamma(2^{h-1})} \sin^{2^{h-1}-1}(t) \cos^{2^{h-1}-1}(t)$. 
\end{Lemma} 
\begin{proof} 
The result follows from the distribution of the angles $\theta_{i}$ at height $h$ computed in Lemma \ref{l9} and the independence of the angles $\theta_{i}$ established in Lemma \ref{l10}.   
\end{proof}

\subsubsection{Correctness of the quantum algorithm} 

We have described the quantum circuits implementing the steps of Algorithm \ref{Brownian}, we are now ready to show that the algorithm produces the quantum states corresponding to superpositions 
over Brownian paths as claimed. 
\begin{Lemma} \label{correct} 
The output of Algorithm \ref{Brownian} is the quantum state  $\ket{B_{L} } = \frac{1}{Z_{3}} \sum_{i \in [T]} B_{L} (i \pi /T) \ket{i}$ representing the Brownian bridge with $T$ time steps obtained by truncating the Wiener series to $L$ terms. 
\end{Lemma} 
\begin{proof} 
Algorithm \ref{Brownian} outputs the quantum state obtained by applying the discrete sine transform to the state $\frac{1}{ Z_{3}} \sum_{k \in [L] } \frac{ a_{k \oplus j} }{ k} \ket{k}  \ket{0^{\log T - \log L}}$ where $j$ is the 
result of the measurement in step $3$. The $L$ dimensional vector $\vec{a} = (a_{1}, a_{2}, \cdots, a_{L})$ has i.i.d. coordinates distributed according to $N(0,1)$. The i.i.d. property is preserved if an arbitrary permutation 
$\sigma \in S_{L}$ is applied to $\vec{a}$. 

For all $j \in [L]$ the mapping $a_{k} \to a_{k \oplus j}$ is a permutation as $k_{1} \oplus j = k_{2} \oplus j$ implies that $k_{1}= k_{2}$. The vector  $\vec{a_{j} } = (a_{1\oplus j}, a_{2 \oplus j}, \cdots, a_{L \oplus j})$
therefore has the same distribution as $\vec{a}$ for all $j$. By Wiener's Theorem, it follows that the discrete sine transform applied to $\frac{1}{ Z_{3}} \sum_{k \in [L] } \frac{ a_{k \oplus j} }{ k} \ket{k}  \ket{0^{\log T - \log L}}$ 
produces the state $\ket{B_{L} } = \frac{1}{Z_{3}} \sum_{i \in [T]} B_{L} (i \pi /T) \ket{i}$.

\end{proof} 

\subsection{Coherent analog encoding for Brownian motion} 
The coherent analog encoding for the Brownian motion is a superposition over the analog representation of the corresponding trajectories, along with the garbage register that encodes the randomness used for generating these trajectories, 
\[ \ket{S(T)}= \displaystyle\sum_{v_1,v_2,\ldots v_T} \sqrt{p_{v_1,v_2,\ldots v_T}} \ket{\psi_{v_1,v_2,\ldots v_T}} \ket{g}  \]
The algorithm for generating the coherent analog encoding for Brownian motion is a two step procedure, the first step creates the angle distributions in Lemma \label{lem:gaussian} such that angles from these distributions when used in a unary data loader generate a Haar random unit vector. As the angle distributions are independent, these distributions are created on independent registers using a unary data loader. After the first 
phase the following quantum state is obtained
\all{ 
 \sum_{\theta_{i}  }  \prod_{i \in [L-1]}  \sqrt{p(\theta_{i})} \ket{\theta_{1}, \theta_{2}, \cdots, \theta_{L-1} }.  
} {step1} 
The second step uses the angle registers to apply the rotations for a unary data loader in order to create the encoding of a $L$ dimensional uniformly random vector on a second register. More precisely, the angles 
$\theta_{i}$ are stored up to some qubits of precisions and controlled rotations conditioned on these bits are applied to get a uniformly random Gaussian vector on the second register. The algorithm can therefore 
be viewed as 'data loading the data loader', the first step generates the angle distributions and the second step uses these angles to prepare a Gaussian state on an independent register. 
In addition, the second step appends an independent register in state $\en{ \sum_{k \in [L]} \frac{1}{ k^{\alpha}}  \ket{ k } }$ with exponent $\alpha$ depending on the Hurst parameter. 
The quantum state is obtained after the second step is,  
\all{ 
\frac{1}{ Z_{1} } \sum_{\theta_{i}  }  \prod_{i \in [L-1]}  \sqrt{p(\theta_{i})} \ket{\theta_{1}, \theta_{2}, \cdots, \theta_{L-1} }  \en{ \sum_{k \in [L]} \frac{1}{ k^{\alpha}}  \ket{ k } } \en{ \sum_{i \in [L]} a_{i} \ket{ i } } , 
} {step2} 
where the normalization factor $\frac{1}{Z_{1}}$ ensures that the last two registers are normalized quantum states with unit norm. 
Using this state in step 3-5 of Algorithm \ref{Brownian} generates an $\epsilon$-approximate coherent analog encoding for Brownian motion, where $\epsilon$ depends on the number of terms $L$ 
retained in the Wiener series. Assuming that the angles distributions for the $\theta_{i}$ in step 1 of the algorithm are generated to a fixed $K$ bits of precision, the number of 
qubits needed is $O(2^{K} L)=O(L)$ and the asymptotic resource requirements for generating the coherent analog encoding of Brownian motion are the same as the requirements for generating 
a single trajectory, 

\begin{theorem} 
There is a quantum algorithm for generating $\epsilon$-approximate coherent analog encodings of Brownian motion with requires $O(L + \log T)$ qubits, 
has circuit depth $\widetilde{O}( \log L + \log T)$ and has gate complexity $O(L + polylog(T))$ for $L=O(1/\epsilon)$. 
\end{theorem} 

\noindent The angle distributions for the $\theta_{i}$ are heavily concentrated at the higher levels of the tree, this can be used to further reduce the resource requirements for near term instantiations of the 
coherent analog encoding preparation procedure for Brownian motion. The quantum Monte Carlo method for estimating expectations over Brownian paths using the coherent analog encoding 
is described in Section \ref{sec:mcmethods}.

\subsection{Quantum runtime and error analysis} \label{error} 
Algorithm \ref{Brownian} offers a potentially significant speedup over classical algorithms for the problem of preparing an analog quantum representations of Brownian motion. The classical algorithm using the discrete Fourier transform has complexity $O(T \log T)$ while the quantum algorithm has complexity $O(L+ \polylog(T))$. The extent of the quantum speedup thus depends on the number of terms $L$ retained in the Wiener series.

We argue that choosing $L$ to be a constant suffices to obtain good approximations of the Brownian path in the $\ell_{2}$ norm. The expected $\ell_{2}$ norm of the Brownian path $E[ \int^{\pi}_{0} B(t)^{2} dt]  = \frac{1}{\pi} ( 1+ \sum_{k\geq 1} \frac{2}{k^{2}}) \leq \frac{1}{\pi} + \frac{\pi}{3}$ 
can be calculated explicitly using Euler's celebrated summation of the series $\zeta(2) = \sum_{k \in \mathbb{N}} \frac{1}{k^{2} } = \pi^{2}/6$. 

A more precise tail bound establishing the asymptotic rate of convergence of the $\zeta(2)$ series can be obtained as follows. 
The $\zeta(2)$ power series truncated at $L$ terms can be approximated by the integral $\int^{\infty}_{L} \frac{1}{ x^{2} } dx = O(1/L)$. 
The expected truncation error $E[ \norm{B(t) - B_{L}(t)}^{2}]$ when the Wiener series is truncated to $L$ terms is $O(1/L)$ and as the 
error is a weighted sum of squares of Gaussian random variables $\norm{B(t) - B_{L}(t)}^{2}$ is concentrated around $O(1/L)$ with high probability. 
Thus $L=O(1/\epsilon)$ terms need to be retained to achieve $\epsilon$-approximate analog encodings (Definition \ref{d2}) for Brownian trajectories
with high probability. The approximation of Brownian motion by a constant number of terms of the Wiener is illustrated in Figure 2.  
\begin{figure} [H] 
\includegraphics[scale=0.35]{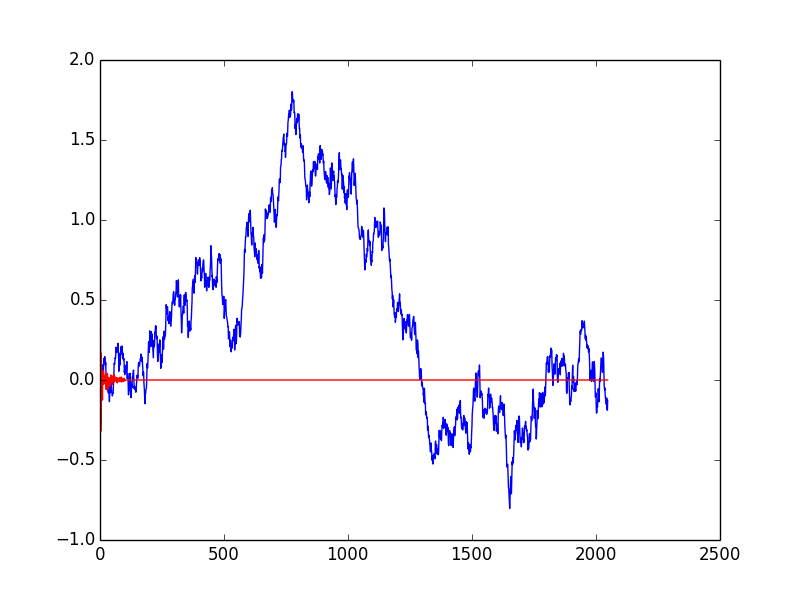}
\includegraphics[scale=0.35]{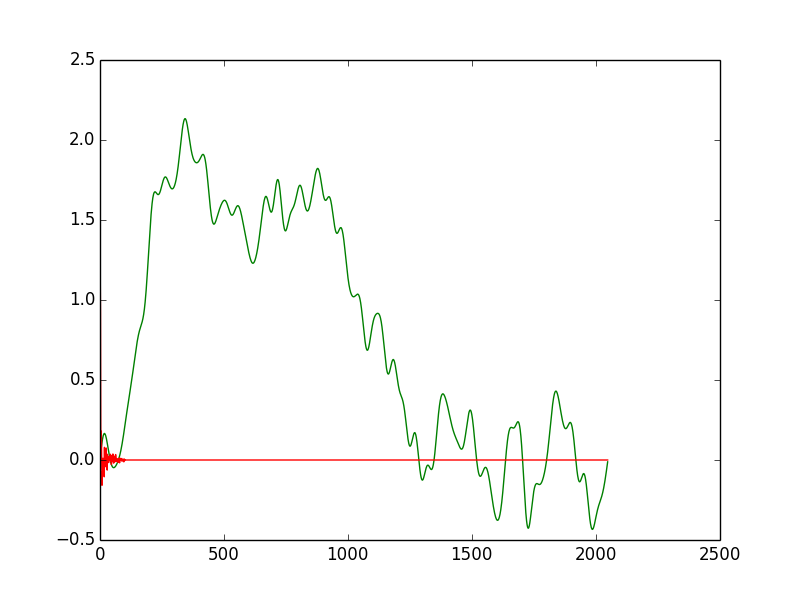}
\caption{(i) Simulation of Brownian bridge using Wiener series for $(N=1000, T=2048)$. (ii) Simulation of Brownian bridge using
truncated Wiener series $(N=100, T=2048)$.} 
\end{figure}

\subsection{Fractional Brownian motion} \label{efbm} 
Algorithm \ref{Brownian} can be used to prepare quantum representations of Fractional brownian motion (Definition \ref{levyfBM}) with arbitrary Hurst parameter $H \in [0,1]$. 
Wiener's Fourier series representation for the Brownian motion can be viewed as arising from the fact that Brownian motion is an integral over white 
noise where the white noise can be described as a stochastic Fourier series with i.i.d. Gaussian coefficients. Fractional Brownian motion in this view 
is a 'fractional' integral over the white noise, where the notion of fractional integral corresponds to the L\'{e}vy or the Mandelbrot definitions of the fBM 
given previously. 

Fractional integrals can be given different formulations, it suffices to determine the fractional integral of $\sin(kt)$ and $\cos(kt)$ to define it for functions having a well 
defined Fourier series. It can be shown that in the Fourier domain, the fractional integral with parameter $\alpha$ corresponds 
to Hadamard product by $k^{-1-\alpha}$ \cite{herrmann2011fractional}. With the interpretation of fractional Brownian motion as a fractional 
integral over white noise it follows that the Fourier coefficients of the fractional Brownian motion with Hurst parameter $H$ decay according to the 
power law with the $k$-th coefficient scaling as $k^{-H- 0.5}$. The quantum algorithm for simulating Brownian motion \ref{Brownian} can be generalized to simulate fBM 
for an arbitrary Hurst parameter by changing the scaling of the power law exponent. 

The number of terms $L$ retained in the stochastic Fourier expansion for the fBM are given by the number of terms required to approximate the power series $\zeta(t)$ for $t \in [1,3]$. 
For $H=0$, the power series $\zeta(1)$ is divergent and thus an arbitrarily large number of terms would be needed, for other values of $H$ the convergence rate is the number of 
terms needed to approximate the zeta power series. Comparing with the previous calculation for Brownian motion, for $H>1/2$ fewer than $200$ terms suffice to approximate the fractional Brownian motion 
up to $99.7\%$ variance while for $H<1/2$ the number of terms $N$ required to achieve this accuracy will be larger. The exact number of terms needed can be calculated from the value of $\zeta(1+2H)$, 
Similar to the case of Brownian motion, the dependence of the approximation error $\epsilon$ in the $\ell_{2}$ norm can be computed by approximating the tail probability for $\zeta(1+2H)$ by the integral $\epsilon = \int_{L}^{\infty} dx/x^{1+2H} = O(1/L^{2H})$. 
The asymptotic error dependence for fractional Brownian motion with $H>1/2$ is better than $O(1/\epsilon)$. For example, only 
$L=O(1/\epsilon^{1/2H})$ terms of the stochastic Fourier series need to be retained to approximate fBM with Hurst parameter $H$ to $\ell_{2}$ error of $\epsilon$. 

The quantum algorithm for simulating fractional Brownian motion with Hurst parameter $H$ is identical to Algorithm \ref{Brownian} with the state 
$\sum_{k \in [L]} \frac{1}{k^{H+ 0.5}}   \ket{k}$ in step 1 of the algorithm. Analogous to the results for Brownian motion ($H=1/2$) in Theorem \ref{thm:BM}, we have the following result on the running time and resource requirements for simulating fractional Brownian motion. 

\begin{theorem} \label{thm:FBM} 
There is a quantum algorithm for generating $\epsilon$-approximate analog encodings of fractional Brownian motion trajectories on $T$ time steps
that requires $O(L + \log T)$ qubits, has circuit depth $\widetilde{O}( \log L + \log T)$ and gate complexity $O(L + polylog(T))$ for $L=O(1/\epsilon^{1/2H})$ where 
$H\in (0,1]$ is the Hurst parameter. 
\end{theorem} 
Fractional Brownian motions for varying $H$ parameters simulated using stochastic Fourier series with power law decay are illustrated in Figure 2. 

\begin{figure} [H] 
\includegraphics[scale=0.35]{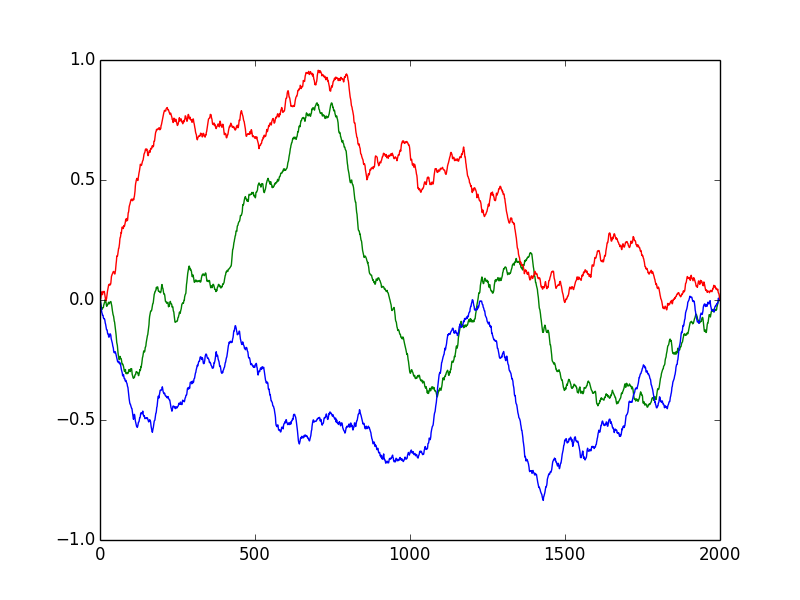}
\includegraphics[scale=0.35]{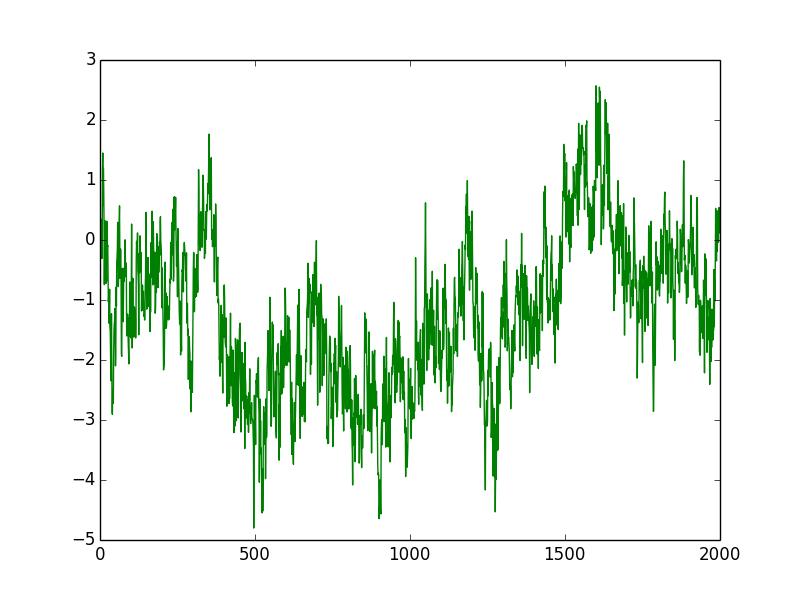}
\caption{(i) Three simulated trajectories for Fractional Brownian motion with Hurst parameter $H=0.8$. (ii) 
A simulated trajectory for Fractional Brownian motion with Hurst parameter $H=0.1$.} 
\end{figure} 

Fractional Brownian motion with $H<1/2$ is an important process in quantitative finance. In \cite{gatheral2018volatility}, it is determined that via estimation of volatility from high frequency financial data that log-volatility time series behave like a fractional Brownian motion, with Hurst parameter of order 0.1. Modeling volatility this way allows one to reproduce the behavior of the implied volatility surface with high accuracy. This result is robust and has been demonstrated with thousands of assets.

\begin{table} \label{tab1} 
\begin{center} 
  \begin{tabular}{|l|c|c|c|} 
   \hline
  &&&\\
   \textbf{ $\epsilon$ } & \textbf{$H=0.5$} & \textbf{$H=0.65$} & \textbf{$H=0.8$}  \\
      &&&\\
    \hline
      &&&\\
      $10^{-2}$ & 100 & 35  & 20 \\
    &&&\\
    \hline
     &&&\\
    $10^{-3}$   & 1000 & 205  & 75 \\
      &&&\\
    \hline 
     &&&\\
      $10^{-4}$ & 10000 & 1200 & 320 \\
    &&&\\
    \hline  
   \end{tabular} 
   \caption {Number of terms in the stochastic Fourier series for fBMs for varying Hurst parameters and error rates. } \label{table1} 
  \end{center} 
   \end{table} 

The truncated stochastic Fourier series captures the large scale variations on the Brownian path while filtering out the high frequency components.For most Monte Carlo estimation applications, the finer scale oscillations on the Brownian path can be safely ignored. The method of retaining only the leading coefficients of the Wiener series to get good approximations is an analogous to principal components analysis where only the leading eigenvalues of the covariance matrix are retained. The generalization of this method to arbitrary stochastic processes is formalized as the Karhunen-Loeve Theorem \cite{H18} in the stochastic processes literature.

\section{Quantum spectral method for stochastic integrals} 
\label{sec:integrals}

\subsection{L\'{e}vy Processes}\label{levy} 
The most general formulation of the quantum spectral method is applicable to stochastic integrals over L\'{e}vy processes. 
We begin with a definition of L\'{e}vy processes and stating some theorems about them. We assume we are given a filtered probability space $(\Omega, \mathbb{P}, \mathbb{F}, \mathcal{F}_t).$ An adapted stochastic process $X$ is called a \textit{L\'{e}vy Process} if it satisfies the following criteria:

\begin{enumerate}

\item  $X_t-X_s$ is independent of $\mathcal{F}_s,$ $0 \leq s \leq t \textless \infty.$

\item $X$ has stationary increments, i.e. $X_t-X_s$ has the same distribution as $X_{t-s},$ $0 \leq s \leq t \textless \infty.$ 

\item $X$ is continuous in probability, i.e. $lim_{t \rightarrow s} X_t=X_s,$ where the limit is taken in probability.

\end {enumerate}

\begin{theorem}[L\'{e}vy-Khintchine]

The following theorem gives the characteristic function of any L\'{e}vy process:

Let $X$ be a L\'{e}vy process with L\'{e}vy measure $\nu.$ Then,

$$E[e^{iuX_t}]=e^{-t \psi(u)},$$

where $\psi(u)$ is given by:

$$ \frac{\sigma^{2} u^{2}}{2} -i\alpha u + \int_{|x| \geq 1} (1-e^{iux}) \nu(dx) +  \int_{|x| \textless 1} (1-e^{iux}+iux) \nu(dx).$$ 

\end{theorem}

It is a consequence of the L\'{e}vy-Khintchine Theorem that any L\'{e}vy process $X$ can be decomposed in the following manner, 
$$X_t=\alpha t +\sigma B_t +Y_t + C_t,$$
with $B$ being a Brownian motion, $Y$ being a pure-jump martingale with jumps bounded in absolute value by $1,$ and $C_t$ being a compound Poisson process, with jumps greater than absolute value $1.$
The approach to quantum simulation of L\'{e}vy processes is to generate first the L\'{e}vy noise and then to apply to it an integration operator, which is then implemented efficiently using the Quantum Fourier transform. 
The method is probabilistic and analysis requires bounds on the the Fourier spectrum of the L\'{e}vy noise, in particular it requires that the ratio of the maximum and the minimum Fourier coefficients is bounded. 

\begin{definition}
Given a L\'{e}vy process $X$, we define \textbf{L\'{e}vy white noise}, $Z,$ to be its generalized derivative:
$$Z_t:=\frac{dX}{dt}-E[X_1].$$
Note that by the properties of L\'{e}vy processes, L\'{e}vy white noise is a zero-mean, stationary process, and $Z_t || Z_s$ for $s \neq t.$
\end{definition}
In order to establish the flatness of the Fourier spectrum of L\'{e}vy noises, we use the Wiener-Khintchine Lemma stated below to 
relate the Fourier coefficients to the auto-correlation function for the process.

\begin{Lemma}[Wiener-Khintchine]

Let $F(\omega)$ be the Fourier Transform of $Z,$ and $G(\omega)$ be the Fourier Transform of its autocorrelation function, $R(\tau).$ 

$$S(\omega)=\lim_{T\to\infty}\frac{\mathbb{E}|F(\omega)|^{2}}{2T}=G(\omega)$$
\end{Lemma}

\begin{proof}

We have $$S(\omega)=\frac{1}{2T}\int_{-T}^{T} \int_{-T}^{T} \mathbb{E}(Z_uZ^{*}_v) e^{-2\pi i(u-v)}dudv=\frac{1}{2T}\int_{-T}^{T} \int_{-T}^{T} R(u-v)e^{-2\pi i(u-v)}dudv.$$  

Let $\tau=u-v.$

Then, the above equals $$ \frac{1}{2T}\int_{-2T}^{2T} R(\tau)e^{-2\pi i(\tau)}(2T-\tau)d\tau$$

Now, we let $T \to \infty,$ arriving at $$S(\omega)=\int_{-\infty}^{\infty}R(\tau)e^{-2\pi i(\tau)}d\tau=\mathcal{F}(R)[\omega]=G(\omega),$$ the Fourier transform of $R.$ 
\end{proof} 

The next claim shows that the Fourier transform  
\begin{claim} 
The power spectrum $S(\omega)$ for L\'{e}vy white noise is flat, that is $S(\omega)= E[(Z_t)^{2}]$ for all frequencies $\omega$.
\end{claim}
\begin{proof} 
To see this, recall the fact that the power spectrum $S$ of L\'{e}vy White noise can be obtained by taking the Fourier transform of its autocorrelation function $R$:
$$S(\omega)= \int_{-\infty}^{\infty}e^{- 2 i \pi\tau\omega }R(\tau)d\tau$$
The autocorrelation function $R(\tau)$ of the process $Z$ is given by $$E[Z_{t+\tau}Z_t].$$ Now, given the properties of L\'{e}vy White noise, we have that  $$R(\tau)=\sigma^{2}\delta_{0}(\tau),$$
where $\sigma^{2}=E[(Z_t)^{2}].$
Therefore, the power spectrum, $S(\omega),$ $\mathcal{F}[\sigma^{2}\delta_{0}]$, which is just the constant $\sigma^{2},$ for all $\omega.$

\end{proof}

Applying Chebyshev's Inequality, we have the following bound on on the supremum of the Fourier coefficients of L\'{e}vy Noise, i.e. $\hat{Z}_k:$

$$P(|\hat{Z}_k| \textgreater M) \leq \frac{\mathbb{E}|\hat{Z}_k|^{2}}{M^{2}}=\frac{C}{M^{2}},$$ 

where the constant $C$ is given by $C=\mathbb{E}[{\hat{Z}_k}^{2}]$

\subsection{Analog encodings for L\'{e}vy processes and stochastic integrals}

We develop a quantum spectral method that can be used to prepare quantum states representing stochastic integrals, the method is applicable to generating analog representations of integrals over L\'{e}vy processes. 
This includes integrals over time or over Brownian paths and It\^{o} processes that are linear combinations of such integrals. The quantum spectral method is obtained by quantizing the classical spectral method for 
generating trajectories for the fractional Brownian motion \cite{dieker2003spectral}.

The spectral method for stochastic processes is based on the observation that the discrete analog of an integral kernel of the form $\int^{t}_{0}  K(t-s) f(s) ds$ corresponds to multiplication of the vector $f(x)$ by a lower triangular Toeplitz matrix in the discrete setting . We recall the definition of Toeplitz matrices and the closely related circulant matrices. 
\begin{definition} 
A Toeplitz matrix $T \in \R^{n \times n}$ is a matrix such that $T_{ij} = f(i-j)$ for some function $f: \R \to \R$, that is the entries of $T$ are constant along the main diagonals. 
\end{definition} 

\noindent The integral $\int^{t}_{0}  K(t-s) f(s) ds$ can be approximated by discretizing the vector $f(S)$ into $T$ time steps and then multiplying the vector $f(s)$ by the lower triangular Toeplitz matrix $(T_{K})_{ij}:= K(i-j)$ if $i>j$ and $0$ otherwise. Further, this is also equivalent to discretizing the Kernel into $T$ time steps and then multiplying by the lower triangular Toeplitz matrix  $(T_{f})_{ij}:= f(i-j)$.

In the quantum setting, we will be computing the 
stochastic integrals $\int^{t}_{0}  K(t-s) \mu(s) ds$ against a L\'{e}vy noise $\mu(s)$. This is achieved by multiplying the state corresponding to the amplitude encoding of the discretized kernel against the Toeplitz matrix $T_{\mu}$. 
, that is $\ket{K} = \frac{1}{ \norm{K}} \sum_{t} K(t) \ket{t}$ by the Toeplitz matrix $T_{\mu}$ generates the amplitude encoding for the function $g(t) = \int^{t}_{0}  K(t-s) \mu(s) ds$.  

Circulant matrices defined below \ref{circ} are matrices generated by the cyclic shifts of some vector $c$. 
The spectral method for simulating stochastic integrals is based on the observation that Toeplitz matrices can be embedded into circulant matrices 
and that circulant matrices are diagonalized by the Fourier transform and  their eigenvalues can be computed explicitly. 

\begin{definition} \label{circ} 
A circulant matrix $C \in \R^{n \times n}$ is a matrix such that $C_{ij} = f((i-j) \mod n)$ for some function $f: \R \to \R$, that is the rows of the matrix $C$ are generated by applying cyclic shifts to its first row. 
\end{definition} 
\noindent A Toeplitz matrix $T_{\mu}$ of dimension $n$ can be embedded into a circulant matrix $C_{\mu}$ of dimension $2n$ as follows, 
\all{ 
C_{\mu} = \left( 
\begin{matrix} 
& T_{\mu} & T_{\mu}^{'} \\
& T_{\mu}^{'} & T_{\mu} 
\end{matrix}  \right )
} {circmult} 
where $T_{\mu}^{'}= (T_{\mu^R})^{T}$ where $T_{\mu^R}$ is the the reversed Toeplitz matrix with first column given by the reverse $\mu^{R}$ of $\mu$. (The reverse $x^{R}$ for $x \in \R^{n}$ is the vector with entries $(x^R)_{j} = x_{n-j}$ of the first column of $T_{x}$). 

It is well known that circulant matrices are diagonalized by the Fourier transform, that is a circulant matrix $C$ has a factorization of the form $C=U \diag(\widehat{c}) U^{-1}$ where $U$ is the unitary matrix 
for the quantum Fourier transform and $\widehat{c}= \QFT( C e_{1} )$ is the Fourier transform of the first column of $C$. The eigenvalues of the matrix $C_{K}$ are thus determined by taking the Fourier transform of the 
first column which by construction is the vector $(K, 0^{T})$.

\begin{algorithm} [H]
\caption{Quantum spectral method for stochastic integrals.} \label{integral} 
\begin{algorithmic}[1]
\REQUIRE Number of time steps $T$ assumed to be power of $2$. Circulant matrix $C_{\mu}$ of dimension $2T$ corresponding to the L\'{e}vy noise $d\mu_{s}$ that is being integrated against kernel $K$.  
\ENSURE A quantum state representing the stochastic integral $F(t) = \int^{t}_{0} K(s, t) d\mu_{s} $ discretized to $T$ steps. 
\STATE Compute the Fourier coefficients $b':= \QFT ((\mu, 0^{T}))$ for the L\'{e}vy noise $d\mu_{s}$ and let $b$ be the truncation to $L$ largest Fourier coefficients. 
\STATE  Prepare a $T$ dimensional amplitude encoding for the kernel $\ket{K}$ and append an extra qubit to get $\ket{K, 0}$. Preparation of $\ket{K}$ may be carried out 
using the unary data loader and the unary to binary convertor in appendix $C$. 
\STATE  Apply the $2T$ dimensional inverse quantum Fourier transform circuit to $\ket{K, 0}$ to ket $\ket{\widehat{K}} = \sum \widehat{k_{i}} \ket{i}$. 

\STATE Append an extra qubit and apply controlled rotations to obtain the state $\sum \widehat{k_{i}} \ket{i} ( \frac{b_{i}}{ \max_{i} b_{i} }  \ket{0} + \sqrt{1-\left(  \frac{ b_{i}} {\max_{i} b_{i} } \right ) ^{2}}  \ket{1})$ and postselect on outcome $\ket{0}$. If outcome $1$ is obtained repeat step 2-4.

\STATE Apply the $2T$ dimensional quantum Fourier transform circuit to obtain an amplitude encoding of  $\int^{t}_{0} K(s, t) d\mu_{s} $ concatenated with its reversal. 
\STATE Measure the auxiliary qubit added in step 2, if $0$ then we have amplitude encoding of  $\int^{t}_{0} K(s, t) d\mu_{s} $if $1$ then we have the reversed amplitude encoding.

\end{algorithmic}
\end{algorithm}

The quantum representation of the stochastic integral $\int^{t}_{0}  K(t-s) \mu(s) ds$ is obtained by multiplying the initial state $\ket{ (K(s), 0^{T})} $ by the matrix $C_{\mu}$. 
Multiplication by the matrix $C_{mu}$ can be in turn implemented using the spectral decomposition $C_{mu} =U \diag(\widehat{\mu}) U^{-1}$, the unitaries $U$ correspond to the quantum Fourier 
transform while the multiplication by the diagonal matrix can be implemented probabilistically using a post-selection step similar to that used in the HHL algorithm.

The algorithm for preparing the representation of the stochastic integral is given as Algorithm \ref{integral}. It is described for the case of integrals over a L\'{e}vy noise $d\mu_{s}$. In particular, integrals over the Brownian 
motion can be obtained by taking $\mu_{s}= dB_{s}$. We next establish the correctness of the algorithm and bound its success probability, 

\begin{theorem} 
Algorithm \ref{integral} generates the amplitude encoding of $\int^{t}_{0} K(s, t) d\mu_{s} $, requires $O(T \log T)$ resources and succeeds with probability at least $1/C^{2}$ where $C=  \frac{  \min_{i} b_{i}}{ \max_{i} b_{i} }$ is the 
ratio of the maximum and minimum Fourier coefficients of the L\'{e}vy noise $\mu_{s}$. 
\end{theorem} 
\begin{proof} 
We argue that Algorithm \ref{integral} implements correctly the multiplication of the discretized kernel by the circulant matrix $C_{\mu}$, that is step 6 generates the amplitude encoding of the result of the following matrix multiplication, 
\al{ 
\left( 
\begin{matrix} 
& T_{\mu} & T_{\mu}^{'} \\
& T_{\mu}^{'} & T_{\mu} 
\end{matrix}  \right )
\left ( \begin{matrix} 
& K  \\
& 0
\end{matrix}  \right )
}  
The result of this matrix multiplication is an amplitude encoding of  $\int^{t}_{0} K(s, t) d\mu_{s} $ concatenated with its reversal. Measuring the last qubit in step 7 yields either 
the amplitude encoding of  $\int^{t}_{0} K(s, t) d\mu_{s} $ if the outcome is $0$ and the amplitude encoding of the reversal of  $\int^{t}_{0} K(s, t) d\mu_{s} $ if the outcome is $1$. Applying the 
operation $\ket{i} \to \ket{T-i}$ to the amplitude encoding of the reversal recovers the amplitude encoding for $\int^{t}_{0} K(s, t) d\mu_{s} $. 

The matrix multiplication by the circulant matrix is implemented using the relation $C_{\mu} = \QFT \diag(\widehat{\mu}) \QFT^{-1}$ in steps 3-5 of the algorithm, steps 3 and 5 are unitary while step 4 involves post-selection. 
The success probability for the post-selection step 4 is  at least $\frac{  \min_{i} b_{i}}{ \max_{i} b_{i} }^{2} \geq \frac{1}{C^{2}}$. The analysis of the spectrum of L\'{e}vy processes in Section \ref{levy} shows that for most L\'{e}vy processes $C$ is
a constant. The resources required for the algorithm are $O(T \log T)$ to compute the Fourier transform of the L\'{e}vy noise classically and truncate to $L$ terms, the quantum resources needed are $O(L \log L + \log^{2}T)$ operations for steps 2-4 of the algorithm. 

\end{proof}

Note that although more general and applicable to L\'{e}vy processes with flat Fourier spectrum, Algorithm \ref{integral} generates the amplitude encodings of trajectories of L\'{e}vy processes. Generating the coherent amplitude encoding 
for L\'{e}vy processes would  require additional quantum resources compared to Brownian motion as the coefficients in the Fourier expansion of L\'{e}vy process are not independent except for the case of Brownian motion and loading a multi-variate distribution over the angles of a data loader is computationally more expensive than individually loading independent angle distributions as in the case of Brownian motion.

\section{Quantum Monte Carlo methods. }
\label{sec:mcmethods}

Analog representations $\ket{S(T)}$ of the stochastic process are compatible with standard quantum Monte Carlo methods and can be used as part of the simulation circuit (oracle) for estimating a function of the stochastic process using the quantum amplitude estimation algorithm. In this section, we provide a quantum Monte Carlo algorithm to estimate functions of stochastic processes given a coherent analog encoding. The method is described for fractional Brownian motion, but is applicable more generally to stochastic processes with coherent analog encodings. 

The goal of the algorithm is to estimate a degree $d$ function $f: S(t)^d \to \R$ of the stochastic process. A linear function with $d=1$ corresponds to the expectation of an inner product $E_{v(t) \sim S}  \braket{f(t)}{v(t)}$. Time averages over the stochastic process are examples of linear functions with $f(t)$ being a step function. Quadratic functions with $d=2$ correspond to mean square averages and variances of linear functions, and  can be expressed as the inner product of a function with two copies of the trajectory of the stochastic process. The quantum Monte Carlo method is more efficient for low degree functions and for Hurst parameters $H>1/2$ for fractional Brownian motion. 

A general setting for which quantum Monte Carlo methods are applicable is that in which the function to be estimated can be encoded as an amplitude. That is, there is a circuit for unitary $U$ such that, 
\al{ 
U \ket{0} = \alpha \ket{x} \ket{0} + \beta \ket{x^{\perp} } \ket{1} 
}
the quantum amplitude estimation algorithm can then be used to estimate $\alpha$ to additive error $\epsilon$ using $O(1/\epsilon)$ queries, and further low depth variants of amplitude estimation can
obtain a speedup that is proportional to the depth $D$ to which the quantum circuit for $U$ can be run on the quantum hardware. 

The quantum Monte Carlo algorithm is given as Algorithm \ref{QMC} where the goal is to estimate either $E_{B_{H}(t)} [\frac{ \braket{f(t)}{B(t)} } { \norm{B(t)} } ]$ or $E_{B'_{H}(t)} [\braket{f(t)}{B(t)}   ]$  where the expectation 
is over fractional Brownian motion trajectories and over trajectories $B'_{H}(t)$ of bounded norm in the second case. Estimating $E_{B_{H}(t)} [\frac{ \braket{f(t)}{B(t)} } { \norm{B(t)} }  ]$ requires less quantum resources but is likely to have an analytic closed form for simple functions $f$ as this is equivalent to computing the Fourier transform for $f$ and computing an inner product with the Wiener series in the Fourier domain. Estimating $E_{B_{H}(t)} [ \braket{f(t)}{B(t)} ]$ requires an additional norm computation and conditional rotation step in Algorithm \ref{QMC}, however the estimate produced does not have a closed form solution as the process in addition post-selects over Brownian paths of a certain norm, thus additional eliminating the additional structure in the Fourier domain arising from the Wiener series expansion.

\begin{algorithm} [H]
\caption{Quantum Monte Carlo algorithm with coherent analog encodings.} \label{QMC} 
\begin{algorithmic}[1]
\REQUIRE Hurst parameter $H$, number of time steps $T$ and terms $L$ in the Wiener series, a function $f: B_{H}(t) \to \R$, quantum circuit $V$ such that  $V \ket{0^{\log l}} = \ket{f} = \frac{1}{ \norm{f} } \sum_{t \in [T]} f(t) \ket{ t }$, the $\norm{f}$ is assumed to be a constant, power law parameter $\alpha=H+0.5$, an upper bound $B_{max}$ on the norm of the Brownian paths considered by the algorithm. 
\ENSURE An additive error $\epsilon$ estimate for the expectation $E_{B_{H}(t)} [\frac{ \braket{f(t)}{B_{H}(t)} } { \norm{B_{H}(t)} }    ]$ where the expectation is over Fractional Brownian motion paths with Hurst parameter $H$ or an estimate for $E_{B'_{H}(t)} [\braket{f(t)}{B_{H}(t)}   ]$ where the Fractional Brownian motion paths are have norm at most $B_{max}$. 

\STATE  Prepare independently the angle distributions and the Gaussian states to obtain the quantum state in equation \eqref{step2}, 
\al{ 
\frac{1}{ Z_{1} } \sum_{\theta_{i}  }  \prod_{i \in [L-1]}  \sqrt{p(\theta_{i})} \ket{\theta_{1}, \theta_{2}, \cdots, \theta_{L-1} }  \en{ \sum_{k \in [L]} \frac{1}{ k^{\alpha}}  \ket{ k } } \en{ \sum_{i \in [L]} a_{i} \ket{ i } }.  
} 
\STATE Apply CNOT gates to map $\ket{k} \to \ket{ k \oplus i } = \ket{j}$ to obtain, 
\al{ 
\frac{1}{ Z_{1} } \sum_{\theta_{i}  }  \sqrt{p(\theta) } \ket{\theta_{1}, \theta_{2}, \cdots, \theta_{L-1} } \ket{j}  \en{ \sum_{k \in [L]} \frac{a_{k \oplus j} }{ k^{\alpha}}  \ket{ k  } }.
}  

\STATE Apply step 4 of Algorithm \ref{Brownian} so that the last register has $\log T$ qubits and apply $V^{-1} (DST)$ on the last register containing the Brownian path to obtain the state, 
\all{ 
\en{ \sum_{\theta_{i}  }  \sqrt{p(\theta) } \ket{\theta_{1}, \theta_{2}, \cdots, \theta_{L-1} } \ket{j} }  \en{ \frac{ \braket{B_{H, j, \theta}(t)}{f(t)} } { \norm{B_{H, j, \theta}(t)} \norm{f} }  |0^{\log T} >  + \beta \ket{0^{\perp} }}.  
}  {fifteen} 
where the normalization factor $1/Z_{1}= 1/\norm{ B_{H, j, \theta}(t)}$ for  $\norm{ B_{H, j, \theta}(t)}= \en{ \sum_{k}  a_{k \oplus j}^{2} / k^{2\alpha}}^{1/2}$ has been explicitly included. 
\STATE Additional step: Compute $\norm{ B_{H, j, \theta}(t)}= \en{ \sum_{k}  a_{k \oplus j}^{2} / k^{2\alpha}}^{1/2}$ in an auxiliary register, append an extra qubit, apply a conditional rotation depending on the norm 
and uncompute the norm to obtain, 
\all{ 
\en{  \sum_{\theta_{i}  }  \sqrt{p(\theta) } \ket{\theta_{1}, \theta_{2}, \cdots, \theta_{L-1} } \ket{j} }  \en{ { \frac{ \braket{B_{H, j, \theta}(t)}{f(t)} } { \norm{B_{H, j, \theta}(t)} \norm{f} } }   \ket{0^{\log T} }  + \beta \ket{0^{\perp} } } \en{ \frac{\norm{B_{H, j, \theta}(t)}}{B_{\max}}    \ket{0 }  + \beta' \ket{1 }} .  
}  {sixteen} 
\STATE Perform quantum amplitude estimation (or a low-depth variant of amplitude estimation) to estimate either: 
\begin{enumerate} 
\item The amplitude for $\ket{0^{\log T} }$ in equation \eqref{fifteen} to additive error $\epsilon \norm{f}$ in order to estimate $E_{B_{H, j, \theta}(t)} [\frac{ \braket{f(t)}{B_{H, j, \theta}(t)} } { \norm{B_{H, j, \theta}(t)} }  ]$ to additive error $\epsilon$.   
\item The amplitude for $\ket{ 0^{\log T+1} }$ in equation \eqref{sixteen} to to additive error $\epsilon B_{\max} \norm{f}$ in order to estimate $E_{B_{H, j, \theta}(t)} [\braket{f(t)}{B_{H, j, \theta}(t)}   ]$ to additive error $\epsilon$.

\end{enumerate}

\end{algorithmic}
\end{algorithm}

\subsection{Proof of correctness}
The correctness proof shows that the estimates produced by algorithm \ref{QMC} are $\epsilon$ close to the true values in additive error. The analysis 
proceeds by analyzing in turn the errors due to truncation of the Wiener series for $L$ terms, tracing out the auxiliary registers and other 
sources of error due to finite precision that are poly-logarithmic. 

\begin{claim} \label{claim1} 
If $L=O(1/\epsilon^{1/2H})$ terms are retained in the stochastic Fourier series for fractional Brownian motion with Hurst parameter $H$, then $E_{B_{H}(t)} [\frac{ \braket{f(t)}{B_{H}(t)} } { \norm{B_{H}(t)} } ]$ is approximated to 
additive error $O(\epsilon)$ for all test functions $f(t)$. 
\end{claim} 
\begin{proof} 
The Fourier series expansion for  fractional Brownian motion with Hurst parameter $H$ is $B_{H}(t) = \sum_{k} \frac{a_{k}}{k^{H+0.5} } \sin(k\theta)$ for i.i.d. Gaussian random variables $a_{k}, k \in \mathbb{N}$.  

The number of terms $L$ that need to be retained in the stochastic Fourier series such that $E[\norm{B_{H}(t) - B_{L}(t) }^{2} ] \leq \epsilon$ can be calculated as follows. 
The tail probability $\sum_{k\geq L+1} a_{k}^{2}/k^{2H+1}$ is approximated by the integral $ \int_{L}^{\infty} dx/x^{1+2H} = O(1/L^{2H})$. Setting the tail probability to $O(\epsilon)$, it follows that $L=O(1/\epsilon^{1/2H})$ terms of the stochastic Fourier series are needed to ensure that $E[\norm{B_{H}(t) - B_{L}(t) }^{2} ] \leq \epsilon$. As $E[\norm{B_{H}(t)}]$ is a constant, we have that $E_{B_{H}(t)} [\frac{ \braket{f(t)}{B_{H}(t)} } { \norm{B_{H}(t)} } ]$ is approximated to additive error $O(\epsilon)$ for all test functions $f(t)$. 
\end{proof} 

The second claim for the correctness of the quantum Monte Carlo method shows that tracing out the $\theta$ and $j$ registers does not affect the expectations. 
\begin{claim} \label{claim2} 
 Tracing out the $\theta$ and $j$ registers does not change the expectation of the quantity being estimated, that is, 
 \al{ 
E_{B_{H, j, \theta}(t)} \frac{ \braket{B_{H, j, \theta}(t)}{f(t)} } { \norm{B_{H, j, \theta}(t)} \norm{f} } = E_{B_{H}(t)} \frac{ \braket{B_{H}(t)}{f(t)} } { \norm{B_{H}(t)} \norm{f} } 
 } 
\end{claim}
\begin{proof} 
The claim is equivalent to showing that after tracing out the $\theta$ and $j$ registers, the states $\sum_{k \in [L]} \tilde{a_{k}}/k^{\alpha}$ in the last register in equation \eqref{sixteen} represent
uniformly random fractional Brownian paths. As the $\theta$ are chosen so tracing out the $\theta$ registers ensures that the $\tilde{a_{k}}$ are i.i.d. Gaussian, it suffices to show that tracing out the $j$ registers, 
 the distribution the  $\tilde{a_{k}}$ remains spherically symmetric, 
\al{ 
\sum_{j} Pr[j] Pr [ \tilde{a} \;|\; \theta, j] =   \frac{1}{ L} \sum_{k, j} \frac{\tilde{a}_{k \oplus j}^{2} }{ k^{2\alpha}} = C\norm{\tilde{a}}^{2} 
} 
It follows that the states in the last register in equation \eqref{sixteen} when $\theta$ and $j$ registers are traced out represent random fractional Brownian paths, so the quantity being estimated by the algorithm is $E_{B_{H}(t)} \frac{ \braket{B_{H}(t)}{f(t)} } { \norm{B_{H}(t)} \norm{f} }$. 
\end{proof} 

\noindent With these auxiliary claims, we can complete the runtime analysis of the quantum Monte Carlo method and resources required for it, 

\begin{theorem} 
Algorithm \ref{QMC} estimates $E_{B_{H}(t)} [\frac{ \braket{f(t)}{B_{H}(t)} } { \norm{B_{H}(t)} }]$ to additive error $\epsilon$ using 
$\widetilde{O}(L + \log T)$ qubits, $\widetilde{O}( (L+ polylog(T) + G)/\epsilon \norm{f})= \widetilde{O}( polylog(T)/\epsilon^{1+1/2H})$ gates where $G$ is the number of gates in the 
circuit $V$ for preparing $\ket{f}$.  
\end{theorem} 
\begin{proof} 
Claim \ref{claim2} shows that  that the amplitude being estimated by algorithm \ref{QMC} is $E_{B_{H}(t)} [\frac{ \braket{f(t)}{B_{H}(t)} } { \norm{B_{H}(t)} } ]$. There are two further sources of error, the first due to the finite precision for generating the distributions on the angles and the second due to 
truncation of the Wiener series to $L$ terms, claim \ref{claim1} shows that the truncation error is $O(\epsilon)$ for $L= O(1/\epsilon^{1/2H})$. Choosing $\log (1/\epsilon)$ qubits to encode each angle further ensures that the errors due to the finite precision of the angle registers is $O(\epsilon)$. 
Thus, the estimate obtained in step 5 of the algorithm is an $O(\epsilon)$ additive error estimate for $E_{B_{H}(t)} [\frac{ \braket{f(t)}{B_{H}(t)} } { \norm{B_{H}(t)} } ]$. 

The number of qubits used is $\widetilde{O}(L + \log T)$ where the $\widetilde{O}$ absorbs potential logarithmic factors due to higher precision 
on the angle registers. The oracle for the amplitude estimation circuit requires $(L+ polylog(T) + G)$ gates and the amplitude estimation algorithm needs to simulate the oracle $1/\epsilon \norm{f}$ times in step 5 of the algorithm to get the desired estimate, and the gate complexity bound follows. 
\end{proof} 
\noindent The classical Monte Carlo method for the task requires resources $\widetilde{O} (T /\epsilon^{2})$ while the quantum Monte Carlo method with using an oracle compiled from classical circuits would require  $\widetilde{O} (T /\epsilon)$.  The gate complexity of the quantum Monte Carlo method using the fBM simulator is $\widetilde{O}( polyLog(T)/\epsilon^{c})$. It is incomparable to the black box quantum Monte Carlo method 
as it achieves an exponential speedup in $T$, the $\epsilon$ dependence is worse. It is more efficient than the classical Monte Carlo method in the regime $c \in [1,2]$, that is for estimating function averages over fractional Brownian paths with $H>1/2$. 

The analysis of the quantum Monte Carlo method covers case 1 where $E_{B_{H}(t)} [\frac{ \braket{f(t)}{B_{H}(t)} } { \norm{B_{H}(t)} } ]$. The 
analysis of the post-selected quantum Monte Carlo method for estimating $E_{B'_{H}(t)} [\braket{f(t)}{B_{H}(t)}   ]$ is not carried out explicitly as it depends on the post-selection procedure, however the post-selected variant is expected to be more useful in practice due to the lack of an analytic closed form solution for the quantity being estimated.

\subsection{Applications to Monte Carlo methods}

In this section, we provide two further examples of applications of the analog encoding of fractional Brownian motion to Monte Carlo methods. The first application is for pricing variance swap options, while the second is for statistical analysis of anomalous diffusion processes. In these end-to-end examples, we harness the $\widetilde{O}( polylog(T)/\epsilon^{c})$ speedup.

\subsubsection{Pricing Variance Swap options} 
In this section, we consider the problem of pricing a variance swap. 

We assume we have a filtered probability space $(\Omega, \mathcal{F}, \mathbb{F}, \mathbb{P}),$ where $\mathcal{F}=(\mathcal{F}_t)_{0\leq t \leq T}$. The filtration $\mathcal{F}$ represents the information available at a given time.

We consider a model such that the price $S$ and volatility $\sigma,$ under a risk neutral measure $\mathbb{Q},$ are given by $$ S_t=e^{\int_{0}^{t}\sigma_sdB_s+\int_{0}^{t}(r-\frac{1}{2}\sigma^{2}_s)ds }$$ 
and $\sigma_t=B^{H}_t$
for a Brownian motion $B$ and an independent  Fractional Brownian motion $B^{H}$ with Hurst parameter $H.$

The log returns, $R_i,$ are given by $$R_i=\ln S_{t_{i+1}}-\ln S_{t_{i}}= \int_{t_i}^{t_{i+1}}(r-\frac{1}{2}\sigma^{2}_s)ds+   
  \int_{t_i}^{t_{i+1}} \sigma_s dB_s \sim N(\int_{t_i}^{t_{i+1}}(r-\frac{1}{2}\sigma^{2}_s)ds , \int_{t_i}^{t_{i+1}} \sigma^{2}_sds)$$.

\noindent A variance swap is an over-the-counter derivative that allows its holder to speculate on the future volatility of the asset price, without any exposure to the asset itself. In such a swap, one party pays amount that is based on the variance of the asset. The other party pays a fixed amount, i.e. the strike price, which is set so that the present value of the payoff is equal to zero.

The realized variance of the asset price over a discretized time interval $0 \leq t_t \leq t_{i+1} \leq T$ is given by $$ \sigma^{2}_{realized} = \frac{A}{n}\sum_{i=1}^{n}{R_i}^{2},$$ where $A$ is an annualization factor, and $n+1$ is the number of observed prices.

The payoff of a variance swap is given $$N_{var}( \sigma^{2}_{realized}-\sigma^{2}_{strike})$$
\noindent In the above, $N_{var}$ is called the variance notional. If we assume that the volatility follows a Fractional Brownian motion, we can use  our analog encoding of Fractional Brownian motion, to compute the strike price, $$\sigma^{2}_{strike}=\mathbb{E}_{\mathbb{Q}}[\sigma^{2}_{realized}| \mathcal{F}_{t_0}]$$ since $$\mathbb{E}_{\mathbb{Q}} 
 (\sigma^{2}_{realized}| \mathcal{F}_{t_0}) = \frac{A}{n}\sum_{i=1}^{n} \mathbb{E}_{\mathbb{Q}} ( {R^{2}_i})=\frac{A}{n}\sum_{i=1}^{n} \mathbb{E}_{\mathbb{Q}} \int_{t_i}^{t_{i+1}} \sigma^{2}_sds. $$  

 \noindent Let $\pi$ denote the projective measurement onto time steps $t_1  \leq t \leq t_n=T.$ The expectation value of $\pi$ is equal to $\sum_{i=1}^{n}{R^{2}_i}.$ and is thus equal to
 $\mathbb{E}_{\mathbb{Q}}][\sigma^{2}_{realized}| \mathcal{F}_{t_0}]$ up to constant factors. \noindent 
 Note that the above quantity can be computed using a quantum Monte Carlo method by modifying Algorithm \ref{QMC} to tag the amplitudes for time steps $1$ to $T$ in step 3. 
 
If the average is computed over all possible sample paths of the prices process $S$, then the realized variance has a closed form solution. Mild post-selection over the paths, for example choosing paths which whose norm is lies in the interval $[B_{min}, B_{max}]$ suffices to ensure that 
$\sum_{i=1}^{n}{R^{2}_i}$ does not have a closed form solution. 

The method described above prices the variance swap option assuming the volatility is an fBM. Pricing an option on realized variance where the \textit{log} volatility is a fractional Brownian motion would require efficient quantum algorithms for generating analog encodings of geometric Brownian motion and more generally exponentiated fractional Brownian motions. Generating such encodings is an important open question for quantum finance. 

\bigskip

\subsubsection{Anomalous Diffusion of Particles} 

A second example is that of single-particle superdiffusion, an example of \textit{}{anomalous diffusion}, in the context of molecular motion. Single-particle tracking is relevant for  particles in microscopic systems as well as animal and human motion. 
Anomalous diffusion is common in (super)crowded fluids, e.g. the cytoplasm of living cells. We next define anomalous diffusion in terms of the 
mean square averages.

 The time-averaged mean-square-deviation (TAMSD) of a particle is a measure of the deviation of the position of a particle with respect to a reference position over time. Let $X_t$ denote the position of the single particle at time $t$ with $t \in [T]$. The TAMSD is then defined as:

 \begin{equation}\label{tam}
M_{T}(\tau)= \frac{1}{T-\tau} \sum_{j=1}^{T-\tau}(X_{j+\tau}-X_{j})^{2}
 \end{equation} 

The mean-square-displacement (MSD) of a particle with position $X_t$ at time $t$ and with PDF of displacement $P(t,x)$ at time $t$ is given by: 

\begin{equation*} \langle X^{2}(t)\rangle=  \int_{-\infty}^{\infty} x^{2}P(t,x)dx \end{equation*}

For a particle following a fractional Brownian motion trajectory, the mean TAMSD has the following scaling: 
$$\langle M(\tau)\rangle \propto \tau^{2H},$$
where $H$ is the Hurst parameter of the fractional Brownian motion. The constant of proportionality is the diffusion coefficient, $D.$ 
\bigskip
The TAMSD be used to classify anomalous diffusion behaviour of a single particle.  We write $$M_{T}(\tau,D,H)$$ to make explicit its dependence on the Hurst parameter $H$ as well as $D.$  We assume that $D$ is known. Calculating the TAMSD of the particles is telling of their diffusive behaviour; that is, it can distinguish between sub and super-diffusive behaviour. 

 Superdiffusion is salient in the traveling behaviour of humans and spreading of infectious diseases \cite{sikora2017mean}. It corresponds to fBM $H \textgreater \frac{1}{2},$ and long-range correlations between displacements and is thus well modeled by Fractional Brownian Motion with $H>1/2$.  
\bigskip 

 Sikora et. al in \cite{sikora2017mean} proposed a statistical test using the TAMSD to characterize anomalous diffusion. It is known that if the particle trajectory $X_t$ follows an fBM, $(T-\tau)M_{T}$ is distributed as \textit{generalized chi-squared}, that is, as $$Y=\sum_{j=1}^{T-\tau}\lambda_{j}(\tau,D,H)U_j,$$
where the $U_j$ are distributed as i.i.d  $\chi^{2}(1)$ or $\gamma(\frac{1}{2}).$ \bigskip 

The $\lambda_j$ are the eigenvalues of the $(T-\tau) \times(T-\tau)$ covariance matrix of the (Gaussian) vector of fBM increments $[(B^{H}_{1+\tau}-B^{H}_1), (B^{H}_{2+\tau}-B^{H}_{2}),........(B^{H}_{N+\tau}-B^{H}_T)]. $ This covariance matrix, $\Sigma,$ that Toeplitz, has the following entries on the $i$th diagonal $\frac{D}{2}[(i+\tau)^{2H}-2i^{2H}+|i-\tau|^{2H}] $. The $(100-\alpha)\%$ confidence interval for the test statistic $M^{T}(\tau)$ is given by, 

\begin{equation}\label{interval}[\frac{DQ_{\frac{\alpha}{2}}}{T-\tau}, \frac{DQ_{1-\frac{\alpha}{2}}}{T-\tau}].\end{equation} 

\noindent In the above, the $Q_{\frac{\alpha}{2}}$ are quantiles such that $Pr(Y \textless Q_{\frac{\alpha}{2}}) \textless \frac{\alpha}{2} $ and $Q_{1-\frac{\alpha}{2}}$ is such that $Pr(Y \textgreater Q_{1-\frac{\alpha}{2}}) \textless \frac{\alpha}{2}.$ The null hypothesis $H_0$ such that $X_t$ follows an fBM with given Hurst parameter, $H_{test}.$ The alternative hypothesis, $H_1,$ is such that the particle trajectory $X_t$ is an fBM with a different Hurst parameter, or that the trajectory follows a continuous-time random walk. $H_0$ is rejected if the test statistic $M^{T}(\tau)$ falls outside of the above confidence interval.

The authors in \cite{sikora2017mean} use Monte Carlo Methods to estimate the power of this statistical test--that is, the probability that it rejects the null hypothesis, given that the alternative hypothesis is true. The power of the test is defined as, $$\Pr(M^{T}(\tau)) \notin [\frac{DQ_{\frac{\alpha}{2}}}{T-\tau}, \frac{DQ_{1-\frac{\alpha}{2}}}{T-\tau}].        $$ 

\noindent Such a Monte Carlo test hinges on the calculation of the empirical probability that the test statistic $M^{T}(\tau)$ does not lie in the above interval.

We can use our analog encoding of Fractional Gaussian noise (using that its spectrum, $f(\omega),$ is proportional to $\omega^{1-2H}$) with given Hurst parameter $H$ to output the following state: 

\all{ 
\frac{1}{T-\tau} \sum_{j \in [T]}( B^{H}(j+\tau)-B^{H}(\tau)) \ket{j}.
} {eq21}

If we assume the alternative distribution of $X$ is a compound poisson process, a special case of a continous time random walk, we can use algorithm \ref{integral} to obtain as output its amplitude encoding. 

Below is quantum Monte Carlo procedure to calculate the power of the statistical test, based on the method in \cite{sikora2017mean}:

\begin{enumerate}
\item Calculate the eigenvalues $\lambda_{j}$ for the Covariance matrix $\Sigma$. 
\item Using the amplitude encoding for shifted fBM in equation \eqref{eq21}, generate $R$ samples of $$(T-\tau)M^{T}(D, H_{test},\tau).$$ 

\item Using the $R$ above samples, calculate empirical quantiles $Q_{\frac{\alpha}{2}}$ and $Q_{1-\frac{\alpha}{2}}$
and the confidence interval in \eqref{interval}. 
\item Use the analog representation of wither fBM with a different Hurst parameter, or a Compound Poisson process, using algorithm \ref{integral} to simulate the alternative distribution. 

\item Estimate the value of the test statistic $M^{T}(\tau),$ using \eqref{tam} or \eqref{eq21}. 

\item Set a random counter $z$ to 0. If the test statistic from the last step falls out of the interval \eqref{interval} computed in step 3, then add $1$ to $z,$ else, add nothing to it. 

\item Repeat the last 3 steps a total of $L$ times, and the power of the test is given by $\frac{z}{L}.$

\end{enumerate}

Note that we can take a projective measurement, $\pi,$ onto times $t_1  \leq t \leq t_n=T.$ The expectation value of $\pi$ is equal to \begin{equation*}
 M_{T}(\tau)= \frac{1}{T-\tau} \sum_{j=1}^{T-\tau}(X_{j+\tau}-X_{j})^{2},
 \end{equation*} that is, the TAMSD. We can then use a simple modification of the Quantum Monte Carlo algorithm in \ref{QMC}, to compute the expectation of~\eqref{tam}.

 Our analog encoding can be used to distinguish between diffusive regimes where both the $H_{test}$ and $H$ under the alternative hypothesis are $ \approx 0.4$ and above, or if the alternative process is a Compound Poisson process. 
It is an open problem to use our analog encoding for very small $H$ in order to detect subdiffusive regimes (e.g. as is the case with telomere motion). \bigskip

Another characterization of of single particle dynamics is its ergodicity, or lack thereof. Ergodicity is characterised by the equivalence of the MSD and the TAMSD in the limit of long trajectory times. That is $$ \lim_{T \to\infty} M^{T}(\tau)= \langle X^{2}(\tau) \rangle $$

\bigskip

We define

$$\langle M^{N}(\tau)\rangle =\frac{1}{N}\sum_{i=1}^{N}M^{T}_{i}(\tau).$$ That is, the mean of the TAMSD over $N$ trajectories of the process. 

We also define $$\xi=\frac{\langle M^{N}(\tau)\rangle}{M^{N}(\tau)} $$

The ergodocity of a stochastic process is characterized by the Ergodicity-Breaking parameter (EB), defined below: 

\begin{equation*} EB (\tau)=\langle {\xi}^{2} \rangle -1         \end{equation*} 

 \bigskip 

We can use the algorithm in \ref{QMC} via quantum amplitude estimation to estimate the EB for Levy Processes, as well as fBM for $H \textgreater \frac{1}{2}.$

One can in turn use these calculations as benchmarks to characterize observed data, for example, via the distibution of $\xi,$ or via large deviation statistics of $\xi.$

\section*{Acknowledgments}
We thank Iordanis Kerenidis for helpful discussions and an anonymous referee for detailed feedback on the manuscript.
Adam Bouland's contribution to this publication was as a paid consultant and was not part of his Stanford University duties or responsibilities.

\bibliographystyle{alpha} 
\bibliography{b1.bib}

\begin{appendix}

\section{Quantum circuit for the discrete sine transform.} \label{dst}

We next provide an implementation of the discrete sine transform as a depth $O(\log T)$ unitary matrix. The discrete sine transform is closely related to the Quantum Fourier Transform (QFT) and the logarithmic depth circuit for it uses a recursive decomposition similar to the QFT. We provide an implementation for the quantum Fourier transform as a quantum circuit of logarithmic depth using Hadamard and phase gates. 

\begin{Lemma} \label{qft} 
The unitary matrix $U_{N}$ for the $N$ dimensional QFT with entries $(U_{N})_{ij} = \omega_{N}^{ij}$ for $N=2^{k}$ and $\omega_{N}= e^{2\pi i/N}$ 
can be implemented as a quantum circuit on $2k$ qubits with depth $O(k \log k)$. 
\end{Lemma} 
\begin{proof} 
We give a recursive description of the circuit for the quantum Fourier transform. The base case is $N=2$, for this case $U_{2}$ is the Hadamard gate. 
Let $x= (x_{0}, x_{2}, \cdots, x_{N-1})$ be the input vector and let $x_{o} = (x_{1}, x_{3}, \cdots, x_{N-1})$ and $x_{e} = (x_{0}, x_{2}, \cdots, x_{N-2})$ be the even and 
odd components of $x$. Then the quantum Fourier transform satisfies the following recurrence for all $j \in [N/2]$, 
\all{ 
(U_{N} x)_{j} = (U_{N/2} x_{e})_{j} + \omega_{N}^{j}  (U_{N/2} x_{o})_{j}  \notag \\
(U_{N} x)_{j+N/2} = (U_{N/2} x_{e})_{j} - \omega_{N}^{j}  (U_{N/2} x_{o})_{j} 
}  {recQFT} 
These recurrences allow us to give a recursive description for the quantum Fourier transform. Let $H_{i}$ denote the Hadamard gate applied to the $i$-th qubit. 
Let $CZ_{k,l} (\omega_{N}^{j})$ be the controlled phase shift with qubit $k$ acting as control qubit, with the phase gate 
$Z= \begin{pmatrix} 1 &0 \\ 0 &\omega_{N}^{j} \end{pmatrix}$ being applied to qubit $l$. Then, 
\al{ 
U_{N} =  \sigma H_{k}  CZ_{k,1} (\omega_{N}^{2^{k-2}}) CZ_{k,2} (\omega_{N}^{2^{k-1}}) \cdots CZ_{k,k-1} (\omega_{N})  U_{N/2}
} 
Let us establish the correctness of this recursive decomposition. 
After the application of $U_{N/2}$ (on qubits 1 through $k-1$) on input state $\ket{x}$, the quantum state $\sum_{j} (U_{N/2} x_{e})_{j} \ket{ j, 0} +  (U_{N/2} x_{o})_{j} \ket{j, 1}$ is obtained. 
The product of the phase gates $\Pi_{j \in [k]} CZ_{k,j} (\omega_{N}^{2^{k-j}})$  controlled on qubit $k$ transforms this state to $\sum_{j} (U_{N/2} x_{e})_{j} \ket{ j, 0} +  \omega_{N}^{j} (U_{N/2} x_{o})_{j} \ket{j, 1}$. From equation \eqref{recQFT} it follows that the
application of $H_{k}$ transforms it to $\sum_{j} (U_{N} x)_{j} \ket{ j, 0} + (U_{N} x)_{j+N/2} \ket{j, 1}$. The permutation $\sigma$ is a cyclic shift that moves the $k$-th qubit to the first position, 
this can be implemented by swapping the wires in the circuit. The final state is $\sum_{j} (U_{N} x)_{j} \ket{ 0, j} + (U_{N} x)_{j+N/2} \ket{1, j}$ which is the quantum Fourier transform of $x$.

As the phase gates can be applied in parallel by making $O(\log k)$ copies of the control qubits, the depth of the QFT circuit requires $2k$ qubits and has depth is $O(k \log k)$. 
\end{proof} 
\noindent The discrete cosine and sine transforms are defined as $DCT(x) = \frac{\overline{U}_{N/2} + U_{N/2}}{2} x$ and $DST(x) = \frac{i\overline{U}_{N/2} - iU_{N/2}}{2} x$. 
The QFT circuit can be used to implement the DCT and DST with the same complexity. 

\begin{corollary} \label{c0} 
The unitaries for the discrete  cosine transform and the discrete sine transform can both be implemented as quantum circuits with $2k+1$ qubits and depth $O(k \log k)$. 
\end{corollary} 
\begin{proof} 
The conjugate of the Fourier transform $\overline{U}_{N/2}$ can be applied by conjugating all the controlled phase gates in the QFT circuit. 
Starting with the state $\ket{0} \ket{ U_{N/2} x} + \ket{1} \ket{ \overline{U}_{N/2} x}$ and applying the $iH$ gate the states corresponding to $DCT(x)$ and $DST(x)$ 
are each obtained with probability $1/2$. As the success probability is known exactly, the probabilistic procedure can be made to succeed with probability $1$ using the exact amplitude 
amplification \cite{BHMT00}.  
\end{proof} 

\section{The sum of squares of $k$ Gaussian random variables} \label{gamma} 

\noindent The distribution of the the sum of squares of Gaussian random variables can be expressed in terms of the Gamma function. 
 
 \begin{Lemma} \label{l4} 
The random variable $X^{2}/2$ where $X \sim N(0,1)$ has distribution $\gamma(1/2)$. 
\end{Lemma} 
\begin{proof} 
Let $G(x)$ be the cumulative distribution function for $X^{2}/2$ where $X \sim N(0,1)$, then, 
\al{ 
G(x) = \Pr [ |X| \leq \sqrt{2x}] = \frac{2}{\sqrt{2\pi} } \int^{\sqrt{2x}}_{0 } e^{-t^{2}/2 } dt 
} 
Making the substitution so that $t^{2}/2=y$ and $t dt= dy \Rightarrow dt = \frac{dy}{\sqrt{2y}}$, 
\al{ 
G(x) 
=  \frac{1}{\sqrt{\pi} } \int^{x}_{0 } y^{-1/2} e^{-y } dy 
} 
Thus $G(x)$ is identical to the cdf for the Gamma distribution with $a=1/2$. 
\end{proof} 

The sum of squares $\sum_{i} X_{i}^{2}/2$ where $X_{i}$ are $k$ independent Gaussian random variables has distribution $\gamma(k/2)$. 
This distribution is more commonly named as the $\chi^{2}$ distribution with $k$ degrees of freedom. 
This follows from the next lemma on the additivity of the gamma distribution under convolution. 
\begin{Lemma} \label{l5} 
If $Y_{1} \sim \gamma(a), Y_{2}\sim \gamma(b)$ then $Y_{1} + Y_{2} \sim \gamma(a+b)$. 
\end{Lemma} 
\begin{proof} The density function $F$ for $Y_{1}+ Y_{2}$ is the convolution of the density functions for $Y_{1}$ and $Y_{2}$, that is, 
\al{ 
\Pr[ Y_{1} + Y_{2} = y] = F(y) = \frac{1}{ \Gamma(a) \Gamma(b)} \int_{t} t^{a-1} (y-t)^{b-1} e^{-y} dt 
} 
Introducing variable $z= t/y$ so that $ydz= dt$, 
\al{ 
 F(y) &= \frac{1}{ \Gamma(a) \Gamma(b)} \int_{t>0} z^{a-1} (1-z)^{b-1} y^{a+b-2} e^{-y}  y dz  \nl
 &=  \frac{y^{a+b-1} e^{-y} }{ \Gamma(a) \Gamma(b)} \int_{z>0} z^{a-1} (1-z)^{b-1} dz \nl 
 &=  \frac{y^{a+b-1} e^{-y} }{ \Gamma(a+b)} 
} 
The final step follows from the definition of the $\beta(a,b)$ probability density function $\frac{\Gamma(a)\Gamma(b)}{\Gamma(a+b)} x^{a-1}(1-x)^{b-1}$. We showed 
that the density function for $Y_{1}+ Y_{2}$ is identical to the density function for $\gamma(a+b)$.  
\end{proof}
\noindent The proposition below is an immediate consequence of Lemmas \ref{l5} and \ref{l4}.

\begin{proposition} \label{c1}
The sum of squares $\sum_{i} X_{i}^{2}/2$ where $X_{i}$ are $k$ independent Gaussian random variables has distribution $\gamma(k/2)$. 
\end{proposition} 
\end{appendix} 

\section{The unary to binary conversion circuit} \label{app:c} 
We provide a logarithmic depth circuit that converts the unary amplitude encoding $\ket{x} = \sum_{i \in [n]} x_{i} \ket{e_{i}}$ for a unit vector $x \in \R^{n}, \norm{x}=1$ to the binary encoding 
$\ket{x} = \sum_{i \in [n]} x_{i} \ket{i}$. The unary encoding requires $n$ qubits while the binary encoding uses only $\log n$ qubits, the convertor circuit operates on $n+ \log n$ qubits. 
The action of the unary to binary convertor circuit on the $n+ \log n$ qubits is given as, 
\all{ 
\left (\sum_{i \in [n]} x_{i} \ket{e_{i}} \right )  \otimes  \ket{0^{\log n} }  \to \ket{0^{n}}  \otimes  \left (\sum_{i \in [n]} x_{i} \ket{i} \right )
} {convert} 
The next proposition shows that the unary to binary conversion circuit can in fact be implemented with logarithmic depth. 
\begin{proposition} 
The unary to binary conversion circuit in equation \eqref{convert} can be implemented by a quantum circuit with depth $O(\log^{2}  n)$ and with $O(n \log n)$ gates and with $O(n)$ 
ancilla qubits. 
\end{proposition} 
\begin{proof} 
Without loss of generality, let $n=2^{k}$ be a power of $2$. 
Given the unary encoding $\left (\sum_{i \in [n]} x_{i} \ket{e_{i}} \right )$, the first bit of the binary encoding is given by the parity of the last $n/2$ qubits of the unary representation. 
The parity of $n/2=2^{k-1}$ qubits can be computed by a circuit with $n$ CNOT gates and with depth $2(k-1)$ using $O(n)$ ancilla qubits. The ancilla qubits store the partial parities 
and are erased at the end of the computation. 

Following the computation of the parity, apply controlled swap gates on qubits $(i,i+n/2)$ for the unary representation with the parity qubit acting as control. 
Recall the two qubit swap gate has the following representation in the standard basis, 
\begin{equation} 
SWAP =  \begin{pmatrix} 
&1  &0 & 0 & 0 \\  
&0  & 0 & 1  &0 \\
&0 & 1  & 0   & 0 \\
&0 & 0 & 0     & 1 
\end{pmatrix}. 
\end{equation}  
The effect of the controlled swap gates is to move the index that is equal to $1$ to the first $n/2$ qubits of the unary representation. The last $n/2$ qubits of the unary representation are equal to $0$ following this 
step and are therefore erased at the end of the computation. The multiple controlled swap gates can be applied in parallel on $n/2$ qubits, this can be accomplished using the available $O(n)$ ancilla qubits to copy the parity bits.  

Following these computations, we have computed $1$ bit of the binary representation and are left with a unary encoding of size $n/2$. Continuing iteratively, we obatin unary to binary conversion circuit with total depth 
$\sum_{1\leq i < (k-1)} (k-i) = O(k^{2}) = O(\log^{2}n)$ is obtained. The number of gates needed for computing each bit of the binary representation is $O(n)$, as there are $O(\log n)$ bits in the binary representation the 
total gate complexity is $O(n \log n)$. The maximum number of ancilla qubits $O(n)$ are needed for computing the first bit of the binary representation, subsequent bits require fewer qubits. The claims on the resource requirements for the unary to binary conversion circuit follow. 
\end{proof} 

\end{document}